\documentclass[12pt]{article}
\usepackage{amsmath, amsthm, amssymb, bbm, setspace,bigints}
\usepackage[margin=1 in]{geometry}
\usepackage{caption, enumitem}
\usepackage{subcaption}
\usepackage[toc,page]{appendix}
\usepackage{graphicx, amsfonts, tikz, multirow}
\usepackage{xurl}
\usepackage{natbib}
\usepackage{mathtools}
\usetikzlibrary{bayesnet}
\usepackage{pdfpages, algpseudocode}
\usepackage{epstopdf}
\usepackage{booktabs}
\usepackage{floatrow, lscape}
\usepackage[ruled]{algorithm2e}
\usepackage[colorlinks,citecolor=blue,urlcolor=blue]{hyperref}
\usepackage[utf8]{inputenc}
\usepackage{authblk}
\usepackage{comment}
\usepackage{xr}
\newfloatcommand{capbtabbox}{table}[][\FBwidth]

\newcommand{\norm}[1]{\Vert#1\Vert}

\def\E{{\mathbb{E} }}

\newcommand{\Wtwo}{W_2}
\newcommand{\pistar}{\pi^*_{\varphi}}
\DeclareMathOperator*{\argmax}{arg\,max}

\allowdisplaybreaks
\newtheorem{theorem}{Theorem}[section]
\newtheorem{lemma}[theorem]{Lemma}

\newtheorem{remark}{Remark}

\newtheorem{assumption}[theorem]{Assumption}
\usepackage{mathtools}
\usepackage{cleveref}

\date{}
\title{Two-Loop Stochastic Mirror Langevin Algorithms \\ for Constrained Sampling}
\author{Ruiting Tong\thanks{tong125@purdue.edu}}
\author{Antik Chakraborty\thanks{antik015@purdue.edu}}
\author{Raghu Pasupathy \thanks{pasupath@purdue.edu}}
\affil{Department of Statistics, Purdue University}
\begin{document}
\maketitle
\begin{abstract}
We study the problem of sampling from a target distribution
$\pi(x)\propto e^{-f(x)}$ supported on a convex set
$\mathcal X\subseteq\mathbb R^d$, when the potential $f$ is accessible
only through a stochastic first-order oracle. Mirror Langevin algorithms
provide a natural approach to constrained sampling by transporting the
problem to an unconstrained dual space and discretizing the resulting
Mirror Langevin diffusion. Existing implementations, however, typically
use a fixed discretization step size and consequently retain a
nonvanishing discretization bias at any fixed step size. Moreover, their
direct extension to settings with noisy gradient information entails the challenge of controlling both discretization and stochastic-oracle error. We study a stochastic first-order version of the Mirror Langevin Algorithm
(sFO-MLA) and, as our main contribution, develop a warm-started
two-loop implementation in which an outer loop progressively decreases
the step size while an inner loop runs sFO-MLA (with a fixed step size) for an appropriately
chosen epoch length. The construction provides a principled schedule
linking step sizes and epoch lengths, so that successive epochs warm-start
from increasingly accurate distributions rather than repeatedly paying
the cost of mixing from a cold start. We establish finite-time
Wasserstein guarantees for sFO-MLA that explicitly separate mixing,
Euler--Maruyama discretization, and stochastic-gradient errors. These
bounds yield a fixed-horizon rate of
$\widetilde O(T^{-1/2})$ and show that the two-loop scheme removes the
associated logarithmic penalty, attaining the canonical
$O(T^{-1/2})$ rate under a geometric step-size schedule and
corresponding epoch lengths. We illustrate the methodology in two statistically distinct settings:
large-scale Bayesian mixture-weight inference on the simplex, where
stochasticity arises from data subsampling, and a doubly intractable
undirected graphical model, where stochastic gradients arise from
Monte Carlo estimation.
\end{abstract}
\section{Introduction}

We consider the problem of sampling from a target probability
distribution $$\pi(x)\propto e^{-f(x)}\mathbf 1_{\mathcal X}(x), \mathcal X\subseteq\mathbb R^d,$$
where $\mathcal X$ is convex and the potential $f$ is accessible only
through a \emph{stochastic first-order oracle}. Specifically, we assume
oracle access to $G(\cdot,\zeta)$ satisfying
$\E_\zeta[G(x,\zeta)] = \nabla f(x),
x\in\mathcal X.$ This setting arises naturally in statistical and machine-learning
problems in which exact gradient evaluation is unavailable or
computationally prohibitive.

A canonical statistical example is provided by latent-variable and
hierarchical models~\citep{1977dem,2007gel}, including missing-data
problems. Suppose, for example, that
$\pi(x)=\int \pi(x,z)\,\mu(dz)\propto e^{-f(x)},x\in\mathcal X,$ where $z$ denotes a latent or unobserved variable. Under standard
regularity conditions, Fisher's identity~\citep{1982lou} expresses the
marginal score as a conditional expectation,
$\nabla_x\log\pi(x)
=
\E\!\left[
\nabla_x\log\pi(x,Z)\mid x
\right],$ so that Monte Carlo sampling of the latent variable naturally produces
a stochastic first-order oracle. A related situation arises in Bayesian
models with parameter-dependent intractable normalizing constants
~\citep{wainwright2008graphical}, where the likelihood gradient itself
must typically be estimated by Monte Carlo. Stochastic first-order
information also arises in simulation-based inference
~\citep{cranmer2020frontier}, simulation-based optimization, and
physics-informed learning~\citep{lin2022multi}. Finally, in large-data
settings, even when the full gradient is available in principle,
evaluating it at every iteration may be prohibitively expensive, while
minibatch gradients provide cheaper unbiased
approximations.

The constraint set $\mathcal X$ introduces an additional computational
challenge. Sampling from constrained distributions is a recurring
problem in modern statistics and machine learning, particularly in
Bayesian inference, where posterior samples are used to approximate
posterior expectations, credible regions, and predictive quantities.
Constraints may arise from scientific, structural, or domain
considerations. Examples include simplex constraints in latent-variable
models such as Latent Dirichlet Allocation~\citep{blei2003latent},
polyhedral constraints in graphical models~\citep{yang2013poisson},
and shape constraints in nonparametric regression
~\citep{guntuboyina2018nonparametric}. Such problems have motivated a
substantial literature on constrained sampling, including random-walk
methods such as Hit-and-Run~\citep{lovasz1999hit} and the Ball Walk
~\citep{kannan1997random}, as well as constrained variants of
Hamiltonian Monte Carlo~\citep{betancourt2011nested,pakman2014exact}.

More recently, constrained sampling over convex sets has been studied
through the lens of Langevin dynamics
~\citep{hsieh2018mirrored,zhang2020wasserstein}, building on the
extensive theory of Langevin methods for unconstrained sampling
~\citep{dalalyan_2017,cheng2018convergence,vempala2019rapid}. A natural
approach borrows from constrained optimization
~\citep{nemirovski1983problem,2003becteb}: a suitable mirror function
$\varphi$, defined on the interior of $\mathcal X$, induces the map
$\nabla\varphi:
\operatorname{int}(\mathcal X)\longrightarrow\mathbb R^d,$ which transports the constrained primal space to an unconstrained dual
space. Mirror Langevin diffusion is then defined in the dual
coordinates, and an Euler--Maruyama discretization
~\citep{2007asmgly} yields the Mirror Langevin Algorithm (MLA);
see, e.g., \citet{li2022mirror,srinivasan2024fast} for recent analyses
and entry points.

Stochastic-gradient variants of mirrored Langevin methods have also
been considered previously. In particular, \citet{hsieh2018mirrored}
study stochastic mirrored Langevin dynamics using stochastic-gradient
information. Our objective here is different. We begin with the
Mirror Langevin diffusion and discretization studied in the more recent
MLA literature and ask how stochastic first-order noise interacts with
its finite-step discretization error, and, more importantly, how one
can systematically drive both errors to zero without repeatedly paying
the cost of mixing from a cold start. %Thus, the main contribution of
%this paper is not merely the replacement of an exact gradient by a
%stochastic one, but a warm-started multistage implementation and an
%analysis quantifying when and why it improves on a fixed-step
%implementation.

The work most directly related to this multistage construction is
\citet{rolland2020double}, who introduced the \emph{double-loop}
Unadjusted Langevin Algorithm (DL-ULA). Their method holds the step size
fixed within an inner loop and decreases it geometrically between outer
iterations, thereby producing a multistage decreasing-step
implementation of ULA. They analyze smooth log-concave targets using
exact gradients, without assuming strong log-concavity, and extend the
same construction to constrained sampling through a Moreau--Yosida
approximation (DL-MYULA). Their algorithm propagates a randomly selected
inner iterate and uses a Euclidean clipping step to control tails in
the analysis. For general $L$-smooth log-concave targets, their
$W_2$ analysis yields an iteration complexity
$\widetilde O(Ld^9\epsilon^{-6})$ for achieving error $\epsilon$, while
their constrained DL-MYULA analysis yields
$\widetilde O(d^{3.5}\epsilon^{-10})$.
See also \citet{2011welteh,2016tehthivol} for related discussions of the
tension between discretization bias and mixing under decreasing
step-size schemes.

We develop the same broad multistage idea in a different regime:
constrained sampling through mirror geometry when only stochastic
first-order information is available. The constraint is handled
intrinsically through the mirror map rather than through penalization,
and, crucially, the terminal state of each epoch initializes the next.
Our analysis shows that a fixed-step stochastic first-order MLA has an
$O(\sqrt{\eta})$ error floor, while reaching that floor from a cold
start requires
$O\!\left(\eta^{-1}\log(1/\eta)\right)$ iterations. This observation
suggests a natural warm-started construction: decrease the step size
geometrically across epochs, while running each epoch for only order
$\eta_k^{-1}$ iterations. Under the strongly monotone dual regime used
in existing MLA theory, this scheme attains
$W_2=O(T^{-1/2})$ as a function of total computational effort $T$,
thereby removing the logarithmic cold-start penalty. We further show
that the underlying warm-start mechanism is not tied to exponential
contraction and extends to a generic polynomial-mixing setting. As a matter of terminology,~\cite{rolland2020double} use the term
``double-loop'' for their multistage construction; we retain the term
\emph{two-loop} here for the corresponding outer--inner structure of
sFO-MLA.

\subsection{Summary of insight and contribution}

We list our insight and contribution in roughly decreasing order of importance.

\begin{enumerate}

\item (Warm-started two-loop sFO-MLA.)
We develop a two-loop implementation of stochastic first-order MLA in
which the outer loop selects a decreasing sequence of
Euler--Maruyama step sizes $\{\eta_k\}$ and epoch $k$ runs the
fixed-step chain for $\ell_k$ iterations. The terminal state of epoch
$k$ initializes epoch $k+1$, so that progressively smaller step sizes
are paired with progressively more accurate initial laws.

Under the strongly monotone dual regime inherited from existing MLA
theory, we show that the geometric schedule
$\eta_k=\eta_0\beta^k,\ell_k\asymp \eta_k^{-1},$
attains $W_2(\mu_k,\pi_\varphi^*)=O(T_k^{-1/2}),$
provided the epoch contraction is sufficiently strong relative to the
reduction in the fixed-step error floor; here $T_k$ denotes the
cumulative number of inner iterations. The analysis makes explicit how
warm starts eliminate the logarithmic cold-start penalty of an
optimally tuned fixed-step implementation.

We further isolate the warm-start mechanism from the particular
exponential contraction supplied by strong monotonicity. Under an
abstract polynomial fixed-step mixing condition and a step-size
dependent invariant-measure bias bound, the same geometric
step-size/epoch construction yields an analogous rate whenever the
epoch-wise forgetting dominates.

\item (Finite-time theory for fixed-step sFO-MLA.)
As the analytical foundation for the two-loop construction, we extend
the finite-time MLA analysis of \citet{li2022mirror} to stochastic
first-order oracles. Under the same structural conditions on the dual
drift and mirror diffusion, we derive a Wasserstein bound that
separates three effects: transient mixing, Euler--Maruyama
discretization error, and stochastic-gradient error. In particular,
the stochastic-gradient perturbation preserves an
$O(\sqrt{\eta})$ fixed-step error floor.

This decomposition also quantifies the interaction between step size
and computational horizon. From a cold start, a fixed step size $\eta$
requires $O\!\left(\eta^{-1}\log(1/\eta)\right)$
iterations for the transient term to reach the
$O(\sqrt{\eta})$ floor. Conversely, optimizing the fixed-step bound
over a horizon $T$ gives
$\eta^*(T)\asymp \frac{\log T}{T},
W_2=\widetilde O(T^{-1/2}).$
These calculations expose the logarithmic penalty removed by
the warm-started two-loop construction.

\item (Implementation under distinct stochastic-oracle mechanisms.)
We study two-loop sFO-MLA in two statistically different settings. The
first is Bayesian mixture-weight inference on the simplex in a
large-data regime, where stochasticity arises from minibatching. The
second is a doubly intractable undirected graphical model, where the
likelihood-gradient contribution is estimated by Monte Carlo. The two
examples therefore exercise the algorithm under different sources of
oracle noise and different constraint geometries.

\item (A downstream application to maximum marginal likelihood
estimation.)
Finally, in Section~\ref{sec:mmle} we show how the fixed-step theory can be used when constrained
sampling is an intermediate computational primitive rather than the
final inferential goal. In a latent-variable maximum marginal
likelihood problem, sFO-MLA approximately samples the conditional
latent-variable distribution appearing in Fisher's identity. We derive
conditions coupling the outer stochastic-approximation step size and
the inner sFO-MLA step size that guarantee almost-sure convergence of
the resulting iterates. %The full result is provided in the
%supplementary material, Section~\ref{sec:mmle}.

\end{enumerate}

\section{Preliminaries}
\label{sec:preliminaries}

Let $\varphi:\mathbb R^d\to(-\infty,+\infty]$
be a Legendre mirror map with
$\operatorname{int}(\operatorname{dom}\varphi)=\operatorname{int}(\mathcal X)$,
such that $\nabla\varphi:\operatorname{int}(\mathcal X)\to\mathbb R^d$
is a bijection. Writing $\varphi^*(y)
:=
\sup_{x\in\mathbb R^d}
\{\langle x,y\rangle-\varphi(x)\},$
we then have $(\nabla\varphi)^{-1}(y)=\nabla\varphi^*(y), \, y\in\mathbb R^d.$ For $y=\nabla\varphi(x), x(y):=\nabla\varphi^*(y),$
the constrained target $\pi$ induces the dual target $\pi_\varphi^*
:=
(\nabla\varphi)_\#\pi$
on the unconstrained space $\mathbb R^d$.

Define $g(y):=\nabla f(x(y)),
H_\varphi(y):=\nabla^2\varphi^*(y),$
and $A_\varphi(y)
:=
H_\varphi(y)^{-1/2}
=
\left[\nabla^2\varphi(x(y))\right]^{1/2}.$
The Mirror Langevin diffusion (MLD) is
$dY_\tau
=
-g(Y_\tau)\,d\tau
+
\sqrt{2}\,A_\varphi(Y_\tau)\,dW_\tau,
\, X_\tau=\nabla\varphi^*(Y_\tau),$
whose invariant dual law is $\pi_\varphi^*$ under appropriate regularity
conditions~\citep{zhang2020wasserstein}. Euler--Maruyama discretization
with step size $\eta>0$ gives the Mirror Langevin Algorithm (MLA),
\begin{align}
Y_{t+1}
&=
Y_t-\eta g(Y_t)
+
\sqrt{2\eta}\,A_\varphi(Y_t)N_{t+1},
\nonumber\\
X_{t+1}
&=
\nabla\varphi^*(Y_{t+1}),
\tag{MLA}
\label{eq:mla-short}
\end{align}
where $\{N_t\}_{t\ge1}$ is an i.i.d.\ sequence of
$N(0,I_d)$ vectors.

We write
$\pi^*_{t,\eta,\varphi}$ and $\pi_{t,\eta,\varphi}$ for the laws of
$Y_t$ and $X_t$, respectively, under~\eqref{eq:mla-short}. When the
fixed-step chain admits an invariant law, we denote its dual and primal
invariant distributions by $\pi^*_{\eta,\varphi}$ and
$\pi_{\eta,\varphi}$. Existing analyses of MLA study, among other
quantities, the discrepancy between these fixed-step distributions and
the target as $\eta\downarrow0$
~\citep{zhang2020wasserstein,li2022mirror}.

Because discretization generally introduces bias, one alternative is
to use the MLA transition as a proposal within a
Metropolis--Hastings correction. If $q(\cdot\mid x)$ denotes the
proposal density from $x$, a proposed state $z$ is
accepted with probability
$\alpha_{\mathrm{MH}}(x,z)
=
\min\left\{
1,
\frac{\pi(z)q(x\mid z)}
{\pi(x)q(z\mid x)}
\right\}.$
A Metropolis-adjusted version of MLA is studied in
\citet{srinivasan2024fast}.

In the present paper, the exact drift $g$ is unavailable. Instead, we
assume access to a stochastic first-order oracle
$G(x,\zeta)$ satisfying $\E_\zeta[G(x,\zeta)]
=
\nabla f(x).$
The corresponding dual stochastic-gradient estimator is $\widehat g(y,\zeta)
:=
G(\nabla\varphi^*(y),\zeta),
\quad
\E_\zeta[\widehat g(y,\zeta)]
=
g(y).$ Replacing $g$ in MLA by $\widehat g$ yields the stochastic first-order
Mirror Langevin Algorithm (sFO-MLA):
\begin{align}
\widehat Y_{t+1}
&=
\widehat Y_t
-\eta\widehat g(\widehat Y_t,\zeta_{t+1})
+
\sqrt{2\eta}\,
A_\varphi(\widehat Y_t)N_{t+1},
\nonumber\\
\widehat X_{t+1}
&=
\nabla\varphi^*(\widehat Y_{t+1}).
\tag{sFO-MLA}
\label{eq:FO}
\end{align}
We denote the corresponding laws by
$\widehat\pi^*_{t,\eta,\varphi}$ and
$\widehat\pi_{t,\eta,\varphi}$.

\begin{remark}[Tractability of the proposal density]
\label{rem:proposal_tractability}
Conditional on a realization of the stochastic gradient, the sFO-MLA
proposal has the same Gaussian structure as its exact-gradient
counterpart. Indeed, let $x$ be the current primal state and let $h$
denote a realized value of
$\widehat g(\nabla\varphi(x),\zeta)$. Then
$\widehat Y_{t+1}
\mid
\{X_t=x,\widehat g(\nabla\varphi(X_t),\zeta)=h\}
\sim
N\!\left(
\nabla\varphi(x)-\eta h,\,
2\eta\nabla^2\varphi(x)
\right).$ Hence, for a proposed primal state $z$, $q(z\mid x,h) =
\mathrm N_d\!\left(
\nabla\varphi(z);
\nabla\varphi(x)-\eta h,\,
2\eta\nabla^2\varphi(x)
\right)
\det\!\left(\nabla^2\varphi(z)\right).$ If $Q_x$ denotes the conditional law of the stochastic gradient given
$X_t=x$, then the marginal proposal density is $q(z\mid x)
=
\int q(z\mid x,h)\,Q_x(dh),$
which is generally unavailable in closed form. By contrast, in the
exact-gradient setting of \citet{srinivasan2024fast}, $Q_x$ is a point
mass at $g(\nabla\varphi(x))$, and the proposal density is explicit.
\end{remark}

%For completeness, let
%$\pi^*_{\tau,\varphi}$ and $\pi_{\tau,\varphi}$ denote the laws of
%$Y_\tau$ and $X_\tau$ in MLD; under stationarity their
%limits are $\pi_\varphi^*$ and $\pi$, respectively. Our error metric throughout the paper is $2$-Wasserstein distance. 

%For probability
%measures $\rho,\nu\in\mathcal P_2(\mathbb R^d)$, $W_2^2(\rho,\nu)
%=
%\inf_{\Gamma\in\Pi(\rho,\nu)}
%\int_{\mathbb R^d\times\mathbb R^d}
%\|u-v\|^2\,\Gamma(du,dv),$
%where $\Pi(\rho,\nu)$ denotes the set of couplings of $\rho$ and $\nu$.
%The primal and dual Wasserstein distances are comparable whenever the
%mirror map is bi-Lipschitz. In particular, if $\varphi$ is
%$m_\varphi$-strongly convex and $L_\varphi$-smooth, and $\rho^*=(\nabla\varphi)_\#\rho, \nu^*=(\nabla\varphi)_\#\nu,$
%then $m_\varphi W_2(\rho,\nu)
%\le
%W_2(\rho^*,\nu^*)
%\le
%L_\varphi W_2(\rho,\nu),$
%or equivalently, $\frac{1}{L_\varphi}W_2(\rho^*,\nu^*)
%\le
%W_2(\rho,\nu)
%\le
%\frac{1}{m_\varphi}W_2(\rho^*,\nu^*).$

\section{Fixed-step~\ref{eq:FO}: Finite Time Guarantees}\label{sec:main_results}
We begin by stating conditions on the dual drift and diffusion
coefficients that yield a contractive fixed-step sFO-MLA recursion.

\begin{assumption}[Strong monotonicity of the dual drift]
\label{ass:A1}
There exists $m>0$ such that $\langle g(y)-g(y'),y-y'\rangle
\geq
m\|y-y'\|^2,
\, y,y'\in\mathbb R^d.
$
\end{assumption}

\begin{assumption}[Lipschitz continuity of the dual drift]
\label{ass:A2}
There exists $M>0$ such that $\|g(y)-g(y')\|
\leq
M\|y-y'\|,
\, y,y'\in\mathbb R^d.$
\end{assumption}

\begin{assumption}[Lipschitz mirror diffusion]
\label{ass:A3}
There exists $\alpha\geq0$ such that $\|A_\varphi(y)-A_\varphi(y')\|_{\mathrm F}^2
\leq
\alpha\|y-y'\|^2, \,y,y'\in\mathbb R^d.$
\end{assumption}

\begin{assumption}[Positive effective gap]
\label{ass:A4}
The effective contraction parameter
$\gamma:=m-\alpha$
is strictly positive. We write $\kappa:=\frac{M}{\gamma}.$
\end{assumption}

Assumptions~\ref{ass:A1}--\ref{ass:A4} are stated directly in the
dual space because this is the space in which both Mirror Langevin
dynamics and sFO-MLA evolve. Up to differences in notation and the
normalization of the diffusion coefficient, these are the same
structural assumptions used in the finite-time analysis of Mirror
Langevin algorithms in \citet{li2022mirror} and related work.
Assumptions~\ref{ass:A1} and~\ref{ass:A2} play roles analogous to
strong convexity and smoothness in Euclidean Langevin analysis:
the former provides contraction of the dual drift, while the latter
controls its variation under discretization. We state
Assumption~\ref{ass:A1} as strong monotonicity rather than strong
convexity because the dual drift
$g(y)=\nabla f(\nabla\varphi^*(y))$ need not itself be the gradient
of a scalar potential in the Euclidean dual coordinates.
Assumption~\ref{ass:A3} controls the additional separation created,
under synchronous coupling, by the state-dependent mirror diffusion.
Consequently, $m$ and $\alpha$ have opposing effects, warranting 
Assumption~\ref{ass:A4}.

We emphasize that Assumptions~\ref{ass:A1} and~\ref{ass:A3}, in
particular, are global Euclidean conditions and can be stringent for
mirror maps adapted to constrained domains. For example, the curvature
of barrier-type mirror maps typically becomes increasingly pronounced
near the boundary, making global Euclidean Lipschitz control of the
state-dependent diffusion restrictive. A more intrinsic analysis could
instead measure drift and diffusion variations in the local norms induced
by the mirror geometry, $\|u\|_x^2
:=
\langle u,\nabla^2\varphi(x)u\rangle, \|v\|_{*,x}^2
:=
\langle v,[\nabla^2\varphi(x)]^{-1}v\rangle,$
and replace the global conditions above by suitable relative or local
counterparts. Such a generalization is important, but is not the primary
objective of the present paper. Our main aim is to isolate and develop
the warm-started two-loop principle for stochastic Mirror Langevin
sampling, and the established assumptions above provide a transparent
setting in which to do so and facilitate direct comparison with the
existing MLA literature. Indeed, as shown in Section~\ref{sec:twoloop},
the two-loop mechanism itself does not fundamentally rely on the
exponential contraction supplied by Assumption~\ref{ass:A1}.
We are now ready to state the main fixed step result for sFO-MLA.
%, providing an upper bound for the $W_2$ distance between $\hat \pi^*_{t,\eta,\varphi}$ and the dual target $\pi^*_{\varphi}$. 

\begin{theorem}[Fixed-Step sFO-MLA]
\label{thm:MS-rec}
Suppose Assumptions \ref{ass:A1}--\ref{ass:A4} hold. Let $\{Y_t\}_{t\ge0}$ and $\{\widehat Y_t\}_{t\ge0}$ denote, respectively, the exact-gradient MLA chain~\eqref{eq:mla-short} and the sFO-MLA chain~\eqref{eq:FO}, defined on a common filtered probability space and driven by the same Gaussian innovations. Suppose that $\pi^*_{0,\eta,\varphi} = \widehat\pi^*_{0,\eta,\varphi},$
and that, for every $t\ge0$, $\E\!\left[ \widehat g(\widehat Y_t,\zeta_{t+1}) \,\middle|\, \mathcal F_t \right] = g(\widehat Y_t),$ and 
$\E\!\left[ \left\| \widehat g(\widehat Y_t,\zeta_{t+1}) - g(\widehat Y_t) \right\|^2 \,\middle|\, \mathcal F_t \right] \le \sigma^2d.$
Then there exist constants $c_{\mathrm L},c_{\mathrm{disc}}>0$ such that, if
$0<\eta\le \eta_{\max}^{\mathrm{FO}} \coloneqq \min\left\{ \frac{\gamma}{M^2}, \; c_{\mathrm L} \frac{\gamma^2} {M^2(1+4\alpha)^2} \right\},$
then, for $t\ge0$,
\begin{align}
\Wtwo\!\left(\widehat\pi^*_{t,\eta,\varphi}, \pi_\varphi^* \right)\le{}& e^{-\gamma\eta t} \Wtwo\!\left( \widehat\pi^*_{0,\eta,\varphi}, \pi_\varphi^* \right) + \sqrt{d\eta} \left[ c_{\mathrm{disc}}(1+4\alpha)\kappa + \frac{\sigma}{\sqrt{\gamma}} \right].
\label{eq:main-W2-generic-fo}
\end{align} 
\end{theorem}\vspace{-0.2in}

The bound in~\eqref{eq:main-W2-generic-fo} separates three effects.
The first term is the transient error and decays exponentially at rate
$\gamma\eta$. The second contribution,
$c_{\mathrm{disc}}(1+4\alpha)\kappa\sqrt{d\eta}$, is the
Euler--Maruyama discretization error, while
$\sigma\sqrt{d\eta/\gamma}$ is the additional error induced by the
stochastic first-order oracle. Thus, for a fixed step size, sFO-MLA
mixes to an $O(\sqrt{\eta})$ neighborhood of the dual target rather
than converging exactly to it. Assuming the mirror map $\varphi$ to be $m_\varphi$-strongly convex, we see that $\nabla \varphi^*$ is $1/m_{\varphi}$ smooth. Theorem~\ref{thm:MS-rec} then gives us also the forward implication in the primal space, that is, $W_2(\hat \pi_{t,\eta,\varphi},\pi) \leq \frac{1}{m_{\varphi}}W_2(\hat \pi^*_{t,\eta,\varphi},\pi^*_{\varphi})$, where $\hat \pi_{t,\eta,\varphi} = (\nabla \varphi^*)_{\#}\hat \pi^*_{t,\eta,\varphi}$ is the primal law. 

Theorem~\ref{thm:MS-rec} also quantifies the tradeoff between the
step size and the run length. For a fixed $\eta$, one may ask how long
does it take before the transient error falls below the
$O(\sqrt{\eta})$ error floor. Conversely, for a fixed computational
horizon $T$, one may optimize $\eta$ to balance these two terms. Suppose the common initial law
$\widehat\pi^*_{0,\eta,\varphi}
=\pi^*_{0,\eta,\varphi}=\mu_0^*$
does not depend on $\eta$, and define
$A_0 \coloneqq \Wtwo\!\left(
\mu_0^*,\pi_\varphi^* \right)$ and $B_0 \coloneqq \sqrt{d} \left[ c_{\mathrm{disc}}(1+4\alpha)\kappa + \frac{\sigma}{\sqrt{\gamma}} \right].$ Therefore, letting $\bar\eta := \min\{\eta_{\max}^{\mathrm{FO}},1\}$, Theorem~\ref{thm:MS-rec} implies that 
\begin{equation}\label{gaja}
\Wtwo\!\left(
\widehat\pi^*_{t,\eta,\varphi},
\pi_\varphi^*
\right)
\le
w(t,\eta)
:=
A_0e^{-\gamma\eta t}
+
B_0 \sqrt{\eta}.
\end{equation} for $0<\eta\le\bar\eta$. Considering fixed step-size implementation, the following lemma then characterizes a sufficient burn-in length $T^*(\eta)$ to ensure that the distance between $\widehat \pi^*_{t,\eta,\varphi}$ and the dual target $\pi^*_{\varphi}$ remains $O(\sqrt{\eta})$ thereafter; conversely, the lemma also characterizes the step size $\eta^*(T)$ minimizing the bound
in~\eqref{gaja} for a fixed sufficiently large horizon, in the process identifying the optimal fixed step size when implementing for a predetermined large horizon $T$. 
\begin{lemma}[EM step size and burn-in length]
\label{lem:fixedstepopt}
Suppose the assumptions of Theorem~\ref{thm:MS-rec} hold.

\begin{enumerate}

\item
Suppose $A_0>0$. For a fixed $\eta\in(0,\bar\eta]$, define
\[
T^*(\eta)
:=
\left\lceil 
\frac{1}{\gamma\eta}
\left[
\log\!\left(
\frac{A_0}{B_0\sqrt{\eta}}
\right)
\right]_+
\right\rceil.
\]
Then, for $t\ge T^*(\eta)$, $w(t,\eta)
\le
2B_0\sqrt{\eta}.$ (If $A_0 =0,$ the same conclusion holds with $T^*(\eta)=0$.)

\item
Suppose $A_0>0$. There exists $T_0<\infty$ such that, for every
integer $T\ge T_0$, the function $w(T,\cdot)$ attains its
minimum over $(0,\bar\eta]$ at
\[
\eta^*(T) =  -\frac{1}{2\gamma T}W_{-1}\left(-\frac{B_0^2}{2\gamma T A_0^2}\right)
\]
where $W_{-1}$ denotes the lower real branch of the Lambert
$W$ function~\citep{2020ebe}. Moreover, $\eta^*(T)
=
\frac{1}{2\gamma}
\frac{\log T}{T}
+
O\!\left(
\frac{\log\log T}{T}
\right),$
and $w\!\left(
T,\eta^*(T)
\right)
\sim
B_0
\left(
\frac{1}{2\gamma}
\right)^{1/2}
\left(
\frac{\log T}{T}
\right)^{1/2}.$

\end{enumerate}
\end{lemma}
% \begin{figure*}[t]
%     \centering
%     \begin{minipage}[t]{0.48\textwidth}
%         \centering
%         \includegraphics[width=\linewidth, height =4.2cm]{Current/etaVT_1.png}
%         \caption*{(a) Optimal step size $\eta^*(T)$.}
%     \end{minipage}
%     \hfill
%     \begin{minipage}[t]{0.48\textwidth}
%         \centering
%         \includegraphics[width=\linewidth]{Current/etaVT_2.png}
%         \caption*{(b) Iteration threshold $T^*(\eta)$.}
%     \end{minipage}

%     \caption{
%     Illustration of the relationship between the step size and the iteration
%     horizon following Lemma \ref{lem:fixedstepopt}. Left: the optimal step size $\eta^*(T)$ as a function of $T$.
%     Right: the iteration threshold $T^*(\eta)$ as a function of $\eta$. Here we have set $\gamma = 1$, $A_0 = 2$ and $B_0 = 1$.
%     }
%     \label{fig:eta-T}
% \end{figure*}

%\subsection{Maximum marginal likelihood estimation}\label{sec:mmle}

Thus a cold-start fixed-step implementation optimally chooses
$\eta^*(T)\asymp(\log T)/T$ and achieves
$W_2=\widetilde O(T^{-1/2})$. Equivalently, at a prescribed small
step size $\eta$, the chain requires
$O(\eta^{-1}\log(1/\eta))$ iterations before its transient error
reaches the $O(\sqrt{\eta})$ floor. 

\begin{remark}[sFO-MLA and marginal likelihood estimation]\label{mmle}The fixed-step characterization in Lemma~\ref{lem:fixedstepopt} is
also useful when sFO-MLA is embedded as an inner sampler within a
larger stochastic algorithm. For example, in latent-variable maximum
marginal likelihood estimation, Fisher's identity expresses the
observed-data score as an expectation with respect to a conditional
latent-variable distribution. When exact sampling from this
conditional distribution is unavailable, sFO-MLA can be used to
approximate the required score. 
%The $O(\sqrt{\eta})$ post-burn-in
%Wasserstein error in Lemma~\ref{lem:fixedstepopt} then gives an explicit
%relationship between the outer stochastic-approximation step size and
%the inner sFO-MLA step size. 
We develop this application and prove a
corresponding almost-sure convergence result in the supplementary
material.\end{remark}

The purpose of the fixed-step analysis above is to identify how stochastic first-order noise changes the finite-time error floor and, crucially, how this floor interacts with the run length. The resulting \(O(\sqrt{\eta})\) floor and \(O(\eta^{-1}\log(1/\eta))\) cold-start burn-in reveal inefficiencies that form the starting point for the warm-started two-loop construction developed next.

\section{Two-loop \ref{eq:FO}}\label{sec:twoloop}
For a fixed step size $\eta$, \eqref{eq:main-W2-generic-fo} suggests that the recursion in~\eqref{eq:FO} may not converge
to the dual target $\pistar$ since the
discretization and stochastic-gradient errors in~\eqref{eq:main-W2-generic-fo} persist even as $t \to \infty$ (whereas the transient error decays to zero). Lemma~\ref{lem:fixedstepopt}
further shows that for each fixed $\eta$, there is a threshold $T^*(\eta)$ such that for $t \geq T^*(\eta)$ the transient error is no longer the dominant contribution. This implies that continuing
to run the chain past $T^*(\eta)$ at the same step size $\eta$ yields diminishing benefit, and any further
improvement in terms of reducing the $W_2$ distance from $\pi^*_{\varphi}$ requires reducing $\eta$. This observation motivates  the two-loop sFO-MLA idea (see Algorithm~\ref{alg:imla}) consisting of an outer loop indexed by $k$, and an inner loop indexed by $t$. At outer iteration (or epoch) $k$, Algorithm~\ref{alg:imla} fixes an
Euler--Maruyama step size $\eta_k$ and executes~\eqref{eq:FO} for
$\ell_k$ inner iterations, that is, for an epoch of size $\ell_k$. The terminal state of the $k$-th outer iteration is then used to
initialize the $(k+1)$-th outer iteration, which operates at a smaller step size
$\eta_{k+1}$. This two-loop structure embodies a ``warm start,'' whereby each successive outer iteration begins from an increasingly accurate
distribution rather than restarting from the original initialization. The
outer-loop dual iterates are then mapped back to the
primal via $X_k=\nabla\varphi^*(Y_k)=(\nabla\varphi)^{-1}(Y_k).$

\begin{remark}[Debiasing using the MH filter]\label{rem:debiasingMH}
 The MH filter~\citep{chib1995understanding}, the classical debiasing alternative to the two-loop idea, involves proposing a new state from $q(\cdot \mid x)$ from the incumbent state $x$, and then accepting the proposed state according to the acceptance probability $\alpha_{\mathrm{MH}}(x,\cdot)  = \min\left\{1,\frac{\pi(\cdot)q(x \mid \cdot)}{\pi(x)q(\cdot \mid x)}\right\}$. 
 Computing this ratio presents two challenges in the current stochastic first-order context.  First, since (only) a noisy oracle is available, replacing $\pi(x) \propto e^{-f(x)}$ by $\hat{\pi}(x) \propto e^{-F(x, \zeta)}$ in the MH acceptance ratio $\alpha_{\mathrm{MH}}(x,\cdot)$ will not achieve debiasing. Second, as discussed in Remark~\ref{rem:proposal_tractability}, the marginal proposal density $q(\cdot \mid x)$ under sFO-MLA may not be available. 
 %Indeed, we provide one such example involving doubly intractable Bayesian inference in Section \ref{sec:numerical_experiments}. Moreover, pseudo-marginal MCMC~\citep{andrieu2009pseudo} does not
%Moreover, pseudo-marginal MCMC does not directly apply, since it requires a nonnegative unbiased estimator
%of the unnormalized target density $e^{-f(x)}\mathbf{1}_{\mathcal{X}}(x)$.
\end{remark}
\begin{algorithm}
\caption{Two-Loop sFO-MLA}\label{alg:imla}
\scalebox{0.8}{ 
\begin{tabular} {l}
  
  \textbf{Input:} initial step $\eta_0$; initial iterate $X_0$; step sizes $\{\eta_k, k \geq 0\}$; \\ epoch lengths $\{\ell_k, k \geq 0\}$ \\ 
  \textbf{Output:} sFO-MLA iterates ${X}_k, k=0,1,2,\ldots$\\
  0 Set $Y_0=\nabla\varphi(X_0)$ \\
  1 \textbf{for} $k = 0,1,2, \ldots $ \textbf{do} \hspace{0.5in}{\scriptsize\# outer loop: update EM step size $\eta_k$} \\
  2 ~~Set $\eta_k, \ell_k$ \\
  3 ~~${Y}_{0,k} \leftarrow {Y}_k, \quad {X}_{0,k} \leftarrow {X}_k, $ \\
  4 ~~\textbf{for} $t=0,1,2,\ldots, \ell_k-1$ \hspace{0.25in}{\scriptsize\# inner loop: run sFO-MLA for $\ell_k$ steps}\\
  5 ~~~~Use~\eqref{eq:FO} to obtain ${X}_{t+1,k},{Y}_{t+1,k}$ \hspace{0.5in}{}\\
  6 ~~\textbf{end for} \\
  7~~${Y}_{k+1} \leftarrow {Y}_{\ell_k,k}, \, {X}_{k+1} \leftarrow {X}_{\ell_k,k}$ \\
  8 \textbf{end for}\\
   
 \end{tabular}
 }
\end{algorithm} 

Given the warm-started two-loop framework of
Algorithm~\ref{alg:imla}, the main question is how the step size
$\eta_k$ and epoch length $\ell_k$ should be related so that the
outer-loop laws $\mu_k:=\mathcal L(Y_k)$ converge to $\pi_\varphi^*$,
and at what rate this convergence occurs. Towards facilitating such analysis, we write the one-epoch bound of~\eqref{gaja} as:
\begin{equation}
D_{k+1}
\le
e^{-\gamma\eta_k\ell_k}D_k
+
B_0\eta_k^{1/2},
\, \,
D_k:=W_2(\mu_k,\pi_\varphi^*).
\label{eq:epoch-generic}
\end{equation} We are now ready to state the mixing rate of iterates generated by the two-loop algorithm. 
%under dual strong monotonicity and smoothness restrictions.

\begin{theorem}[Two-Loop sFO-MLA Bounds Under Dual Strong Monotonicity]
\label{thm:two-loop-short}
Suppose the assumptions of the corresponding fixed-step result hold
uniformly over epochs. For $0<\beta< 1$ and $c_0>0$, let $\eta_k=\eta_0\beta^k,\, \ell_k=
\left\lceil\frac{c_0}{\eta_k}\right\rceil$ and define $q_{1/2} \coloneqq e^{-\gamma c_0}\beta^{-1/2}.$
If $T_k:=\sum_{j=0}^{k-1}\ell_j$
denotes the cumulative number of inner iterations, then
$T_k\asymp\beta^{-k}$ and
\begin{equation}
D_k
=
\begin{cases}
O(T_k^{-1/2}), & q_{1/2}<1,\\[1mm]
O\!\left((\log T_k)T_k^{-1/2}\right), & q_{1/2}=1,\\[1mm]
O\!\left(
T_k^{-\gamma c_0/|\log\beta|}
\right), & q_{1/2}>1.
\end{cases}
\label{eq:two-loop-rates}
\end{equation}
In particular, if $c_0>
\frac{1}{2\gamma}|\log\beta|,$
then $D_k=O(T_k^{-1/2}).$
\end{theorem}

% \begin{proof}
% Since $\eta_k\ell_k\geq c_0$,~\eqref{eq:epoch-generic} gives $D_{k+1}
% \leq
% aD_k+B_0\eta_0^{1/2}\beta^{k/2}, a:=e^{-\gamma c_0}.$
% Iterating over $k$ yields $D_k
% \leq
% a^kD_0+
% B_0\eta_0^{1/2}
% \sum_{j=0}^{k-1}a^{k-1-j}\beta^{j/2}.$
% Equivalently, $D_k
% \leq
% a^kD_0+
% B_0\eta_0^{1/2}\beta^{(k-1)/2}
% \sum_{j=0}^{k-1}q_{1/2}^{\,j}, q_{1/2}=a\beta^{-1/2}.$
% The three cases in~\eqref{eq:two-loop-rates} now follow from the
% geometric sum and the fact that $T_k\asymp\beta^{-k}$.
% \end{proof}
%\begin{table}[t]
%\centering
%\small
%\renewcommand{\arraystretch}{1.22}
%\setlength{\tabcolsep}{3pt}
%\scalebox{0.9}{
%\begin{tabular}{llcc}
%\toprule
%\textbf{Method}
%&
%\textbf{Oracle}

%&
%\textbf{Fixed horizon}
%&
%\textbf{Two-loop}
%\\
%\midrule
%MLA
%&
%exact gradient
%&

%$\widetilde O(T^{-1/2})$
%&
%$O(T^{-1/2})$
%\\

%sFO-MLA
%&
%stoch.\ first-order
%&

%$\widetilde O(T^{-1/2})$
%&
%$O(T^{-1/2})$
%\\
%\bottomrule
%\end{tabular}
%}
%\caption{
%Finite-horizon and two-loop rates in the smooth strongly monotone
%regime. Here $T$ denotes the number of Markov-chain iterations, rather than
%the total number of scalar oracle evaluations. In all cases the optimal
%EM step size satisfies
%$\eta^*(T)\asymp(\log T)/T$.
%}
%\label{tab:comparison}
%\end{table}

Four observations about Theorem~\ref{thm:two-loop-short} are salient. First, the two-loop sFO-MLA algorithm is a general paradigm for debiasing iterates generated by  (fixed step) sFO-MLA. Two-loop sFO-MLA is especially useful in settings where the potential $f$ can be observed (only) with a stochastic first-order oracle, rendering classical debiasing alternatives such as MH challenging if not impossible to implement --- see Remark~\ref{rem:proposal_tractability} and Remark~\ref{rem:debiasingMH}, and also discussion in~\cite{2011welteh} and~\cite{2016tehthivol}. Second, if the outer loop step size sequence $\eta_k$ is chosen to decrease geometrically as $\eta_k = \eta_0\beta^k, \beta < 1$, and the diffusive epoch length $\ell_k$ is chosen such that $\ell_k = \lceil \frac{c_0}{\eta_k} \rceil$, then the $W_2$ distance between the outer-loop law $\mu_k$ and the dual target $\pi^*_{\varphi}$ decays as $O(T_k^{-1/2})$ as long as $\beta \in (e^{-2(m-\alpha)c_0},1).$ Since the effective gap $m-\alpha$ is generally unknown, there is a risk that the chosen $\beta \notin (e^{-2(m-\alpha)c_0},1)$ resulting in the slower rates appearing in~\eqref{eq:two-loop-rates}. Such deterioration in rate due to an accidental poor choice of parameter should be reminiscent of corresponding discussions~\citep{2009nemetal} in the context of stochastic optimization. Third, the absence of the logarithmic factor in the favorable two-loop regime is a consequence of warm starts. A fixed-step chain initialized
independently must be run for order $\eta^{-1}\log(1/\eta)$ iterations before its transient reaches its error floor. In contrast, epoch $k$ begins from the endpoint of epoch $k-1$, which is already close to the preceding, larger error floor; consequently, an epoch length of order $\eta_k^{-1}$ suffices to track the successively decreasing floors. Fourth, the statement of Theorem~\ref{thm:two-loop-short} uses the fixed-step bound in~\eqref{gaja}
uniformly across outer iterations, so that the contraction and error
constants do not depend on the law initializing a particular epoch.
This assumption is automatic when the structural conditions underlying
the fixed-step result hold globally. More generally, suppose that epoch
$k$ satisfies $W_2(\mu P_{\eta_k}^{\ell_k},\pistar)
\le
e^{-\gamma_k\eta_k\ell_k}
W_2(\mu,\pistar)
+
B_k\eta_k^{1/2},$ where $P_{\eta_k}$ is the Markov transition kernel of one~\eqref{eq:FO} step in the dual space at fixed step size $\eta_k$,  
and $\gamma_k>0$, $B_k<\infty$ are epoch-dependent constants.
If $\inf_k\gamma_k=\underline\gamma>0,\sup_k B_k<\infty,$
then the conclusions of Theorem~\ref{thm:two-loop-short} remain valid with the effective gap
$m-\alpha$ replaced by $\underline\gamma$. Thus uniformity of the
individual constants is not essential. %; what is needed is a uniform
%positive lower bound on the effective contraction and a uniform upper
%bound on the fixed-step error constant.

Theorem~\ref{thm:two-loop-short} relies on
Assumptions~\ref{ass:A1}--\ref{ass:A4} through the exponential
transient bound in~\eqref{gaja}. In particular, dual strong
monotonicity together with the positive effective gap $m-\alpha>0$ yields
exponential forgetting of the initial distribution. The proof of Theorem~\ref{thm:two-loop-short},
however, reveals that such exponential forgetting is not essential to
the two-loop principle. Instead, what is needed is sufficient contraction
over each epoch relative to the decrease in the fixed-step error
floor. We formalize this observation by replacing
Assumptions~\ref{ass:A1}--\ref{ass:A4} with the following two
conditions.

%We have proved the two-loop sFO-MLA mixing result in Theorem~\ref{thm:two-loop-short} under  Assumptions~\ref{ass:A1}--\ref{ass:A4} which lead crucially to the bound in~\eqref{gaja}. In particular, the dual strong monotonicity in Assumption~\ref{ass:A1} along with the assumed positive gap in Assumption~\ref{ass:A4} lead to the exponentially decaying transient term in~\eqref{gaja}, thereby ensuring that the fixed step chain forgets its initial conditions exponentially fast in the two-loop analysis.  Interestingly, however, the proof of Theorem~\ref{thm:two-loop-short} suggests that two-loop's underlying principle is more fundamental, and unreliant on such exponential forgetting of initial conditions. To make this idea precise, let's make the following two assumptions in lieu of Assumptions~\ref{ass:A1}--\ref{ass:A4}.

\begin{assumption}[Polynomial fixed-step mixing]
\label{ass:A5}
There exist constants $\bar\eta>0$, $C_{\rm mix}\geq 1$, and $q>0$
such that, for every $\eta\in(0,\bar\eta]$, the fixed-step kernel
$P_\eta$ admits a unique invariant measure
$\pi_\eta^*\in\mathcal P_2(\mathbb R^d)$ satisfying
$W_2\!\left(\mu P_\eta^t,\pi_\eta^*\right)
\leq
C_{\rm mix}(1+\eta t)^{-q}
W_2\!\left(\mu,\pi_\eta^*\right)$
for every $\mu\in\mathcal P_2(\mathbb R^d)$ and every
$t\in\mathbb Z_+$.
\end{assumption}

\begin{assumption}[Fixed-step bias]
\label{ass:fixedstep-bias}
Let $\pi^*=\pi_\varphi^*$ denote the dual target.
There exist constants $C_{\rm bias}<\infty$ and $r>0$ such that $W_2\!\left(\pi_\eta^*,\pi^*\right)
\leq
C_{\rm bias}\eta^r, 0<\eta\leq\bar\eta.$
\end{assumption}
Assumptions~\ref{ass:A5}--\ref{ass:fixedstep-bias} are standard
types of conditions in the theory of subgeometrically ergodic Markov
chains and numerical approximations of ergodic diffusions. The polynomial
mixing in Wasserstein distance of Assumption~\ref{ass:A5} can be established, for example, from a
Foster--Lyapunov drift condition of the form
$P_\eta V(y)
\le
V(y)-c\eta V(y)^\theta+b\eta,0<\theta<1,$
together with a suitable local Wasserstein contraction or
$d$-small-set condition; see, e.g.,
\citet{butkovsky2014subgeometric} and
\citet{durmus2016subgeometric}. The particular form in
Assumption~\ref{ass:A5}, in which the transient term is
proportional to the initial Wasserstein error, may alternatively be
verified directly through a polynomial pairwise coupling estimate, something we do not detail here.

The fixed-step bias condition in Assumption~\ref{ass:fixedstep-bias} is a standard invariant-measure
approximation requirement for Euler-type discretizations. Under
sufficient smoothness of the drift and diffusion coefficients, moment
stability, and ergodicity uniformly in the step size, invariant measures
of Euler schemes converge to the invariant law of the underlying
diffusion --- see \citet{talay1990second} and
\citet{mattingly2002ergodicity}. In the stochastic-gradient setting,
unbiased gradient noise with uniformly controlled second moments adds a
further perturbation term, and Wasserstein error bounds of this type are
well established for Langevin schemes with inaccurate or stochastic
gradients; see, for example, \citet{dalalyan2019user} and
\citet{barkhagen2021stochastic}. In our setting these conditions motivate
the bound $W_2(\pi_\eta^*,\pi^*)\le C_{\rm bias}\eta^r,$
with $r=1/2$. 
%We are now ready to state the analogue of
%Theorem~\ref{thm:two-loop-short} under Assumption~\ref{ass:A5} and Assumption~\ref{ass:fixedstep-bias}.

\begin{theorem}[Generic Two-Loop sFO-MLA Bounds]
\label{thm:twoloop-poly}
Suppose Assumptions~\ref{ass:A5}
and~\ref{ass:fixedstep-bias} hold. Let
$\eta_k=\eta_0\beta^k,\ell_k=\left\lceil\frac{c_0}{\eta_k}\right\rceil,k\geq0,$
where $\eta_0\in(0,\bar\eta]$, $\beta\in(0,1)$, and $c_0>0$.
Let $\mu_k$ denote the law of the dual iterate at the beginning
of epoch $k$, so that $\mu_{k+1}
=
\mu_k P_{\eta_k}^{\ell_k},$
and define $D_k:=W_2(\mu_k,\pi^*)$ and the cumulative number of inner iterations $T_k:=\sum_{j=0}^{k-1}\ell_j \asymp\beta^{-k}$. 

Set $\rho
:=
C_{\rm mix}(1+c_0)^{-q},$ assume $\rho<1$, and define $Q_r:=\rho\beta^{-r}.$
Then $D_{k+1}
\leq
\rho D_k+B_r\eta_k^r,B_r=(1+C_{\rm mix})C_{\rm bias},$
and, consequently,
\[
D_k=
\begin{cases}
O(T_k^{-r}),
&
Q_r<1,\\[2mm]
O((\log T_k)T_k^{-r}),
&
Q_r=1,\\[2mm]
O(T_k^{-s}),
&
Q_r>1,
\end{cases}
\]
where $s:=
\frac{-\log\rho}{|\log\beta|}.$ In particular, if $r=1/2$ and $c_0>
\left(C_{\mathrm{mix}}\beta^{-r}\right)^{1/q}-1,$
then $D_k=O(T_k^{-1/2}).$
\end{theorem} 

% \begin{proof}
% By Assumptions~\ref{ass:A5} and~\ref{ass:fixedstep-bias},
% the triangle inequality gives, for any $\mu\in\mathcal P_2(\mathbb R^d)$, $W_2(\mu P_\eta^t,\pi^*)
% \leq
% C_{\rm mix}(1+\eta t)^{-q}W_2(\mu,\pi^*)
% +
% B_r\eta^r,$
% where $B_r=(1+C_{\rm mix})C_{\rm bias}$.
% Applying this bound over epoch $k$ and using
% $\eta_k\ell_k\geq c_0$ yields $D_{k+1}\leq \rho D_k+B_r\eta_0^r\beta^{rk}, \rho=C_{\rm mix}(1+c_0)^{-q}.$
% Iteration gives $D_k
% \leq
% \rho^kD_0+
% B_r\eta_0^r
% \beta^{r(k-1)}
% \sum_{j=0}^{k-1}Q_r^j,Q_r=\rho\beta^{-r}.$
% The three regimes follow from the geometric sum and
% $T_k\asymp\beta^{-k}$.
% \end{proof}

Taken together, Theorems~\ref{thm:two-loop-short}
and~\ref{thm:twoloop-poly} identify the central mechanism behind the
two-loop construction. The geometric decrease in step size progressively
lowers the fixed-step error floor, while warm-starting from the terminal
distribution of the preceding epoch limits the additional effort required
to track these successively more accurate distributions. Importantly,
this mechanism is not tied to exponential mixing: it continues to operate
under polynomial fixed-step forgetting, provided that each epoch contracts
sufficiently relative to the reduction in its error floor. 

%We next examine
%the practical behavior of two-loop sFO-MLA in two settings that are distinguished by the source of oracle noise and the nature of constraint geometry. %The first is a large-scale Bayesian mixture model on the simplex,
%where stochasticity arises through data subsampling; the second
%is a doubly intractable undirected graphical model, where the stochastic
%first-order oracle is generated through Monte Carlo. %Together,
%these examples illustrate the algorithm under different constraint
%geometries and qualitatively different sources of oracle noise.

\section{Numerical experiments}\label{sec:numerical_experiments}
All experiments were performed on a MacBook Air equipped with an Apple M2 processor and 8 GB of unified memory, running macOS 26.6.2 on the ARM64 architecture. 
%The implementations used R 4.5.1, Rcpp 1.1.2, and Apple Clang 14.0.3. The implementations used no parallelization or GPU acceleration, and the BLAS backend was configured to use a single thread.

\begin{figure*}
    \centering

    \begin{minipage}[t]{0.23\textwidth}
        \centering
        \includegraphics[width=\linewidth, height = 4cm]{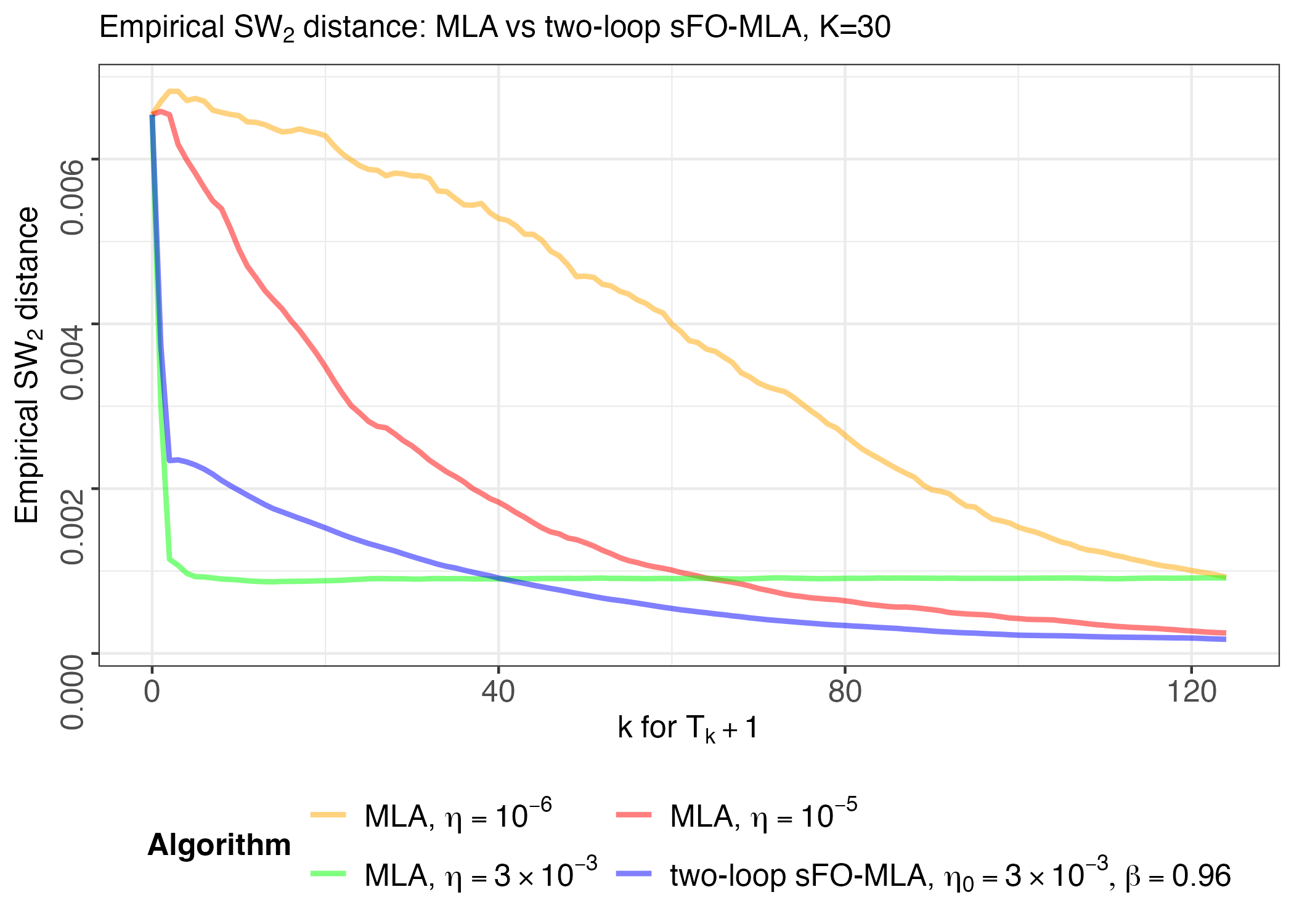}
    \end{minipage}
    \hfill
    \begin{minipage}[t]{0.23\textwidth}
        \centering
        \includegraphics[width=\linewidth, height = 4cm]{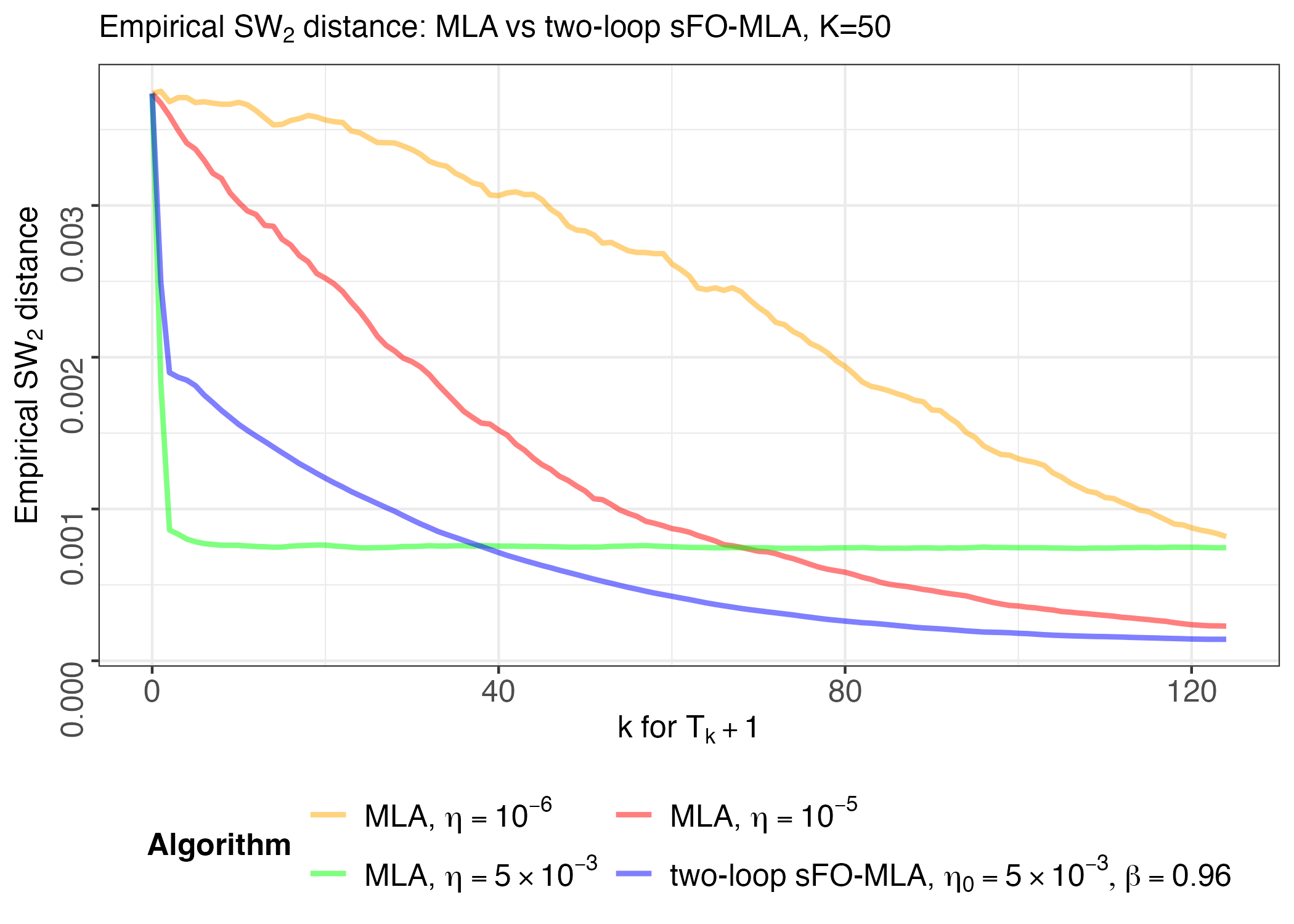}
    \end{minipage}
    \hfill
    \begin{minipage}[t]{0.23\textwidth}
        \centering
        \includegraphics[width=\linewidth, height = 4cm]{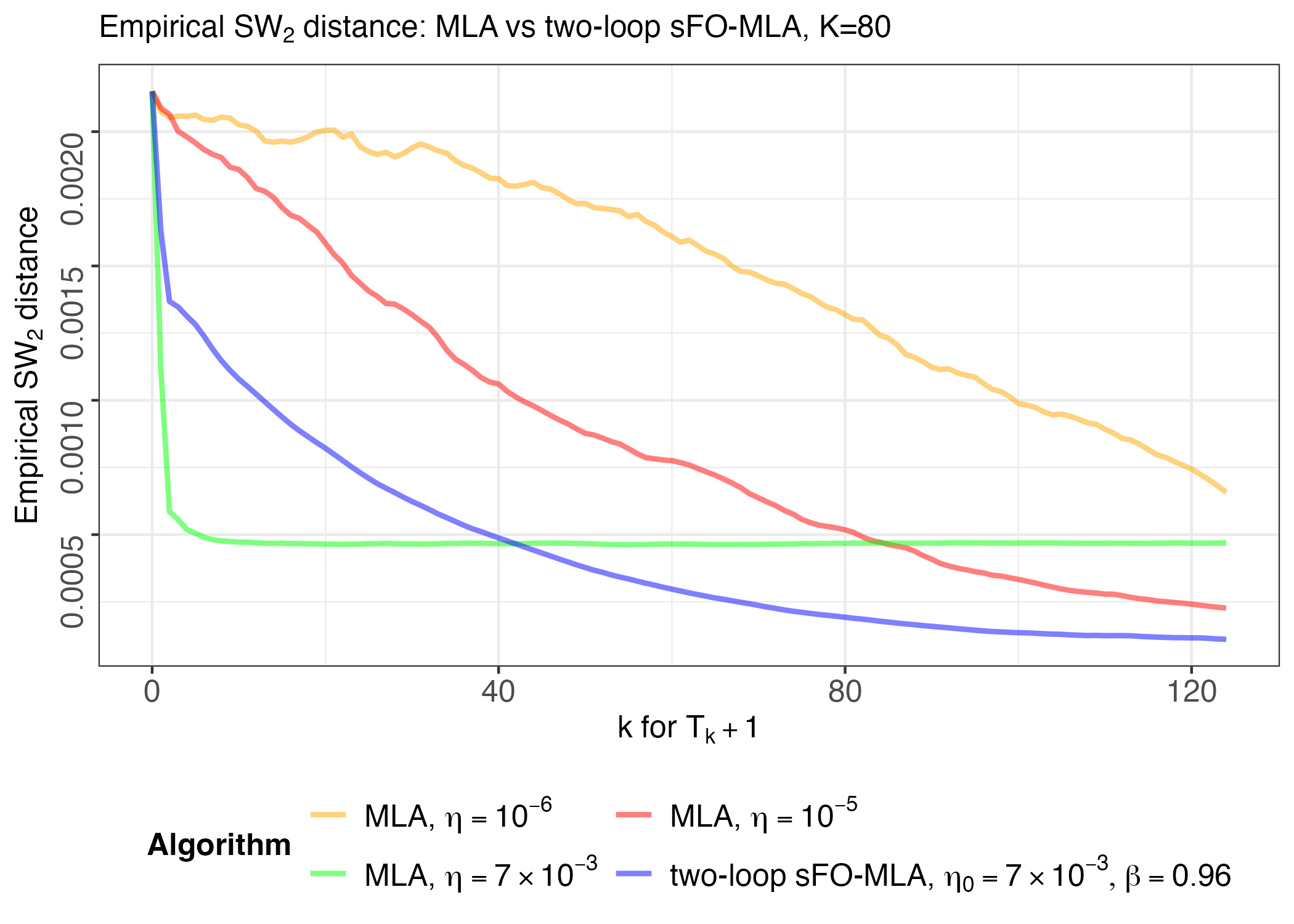}
    \end{minipage}
    \hfill
    \begin{minipage}[t]{0.23\textwidth}
        \centering
        \includegraphics[width=\linewidth, height = 4cm]{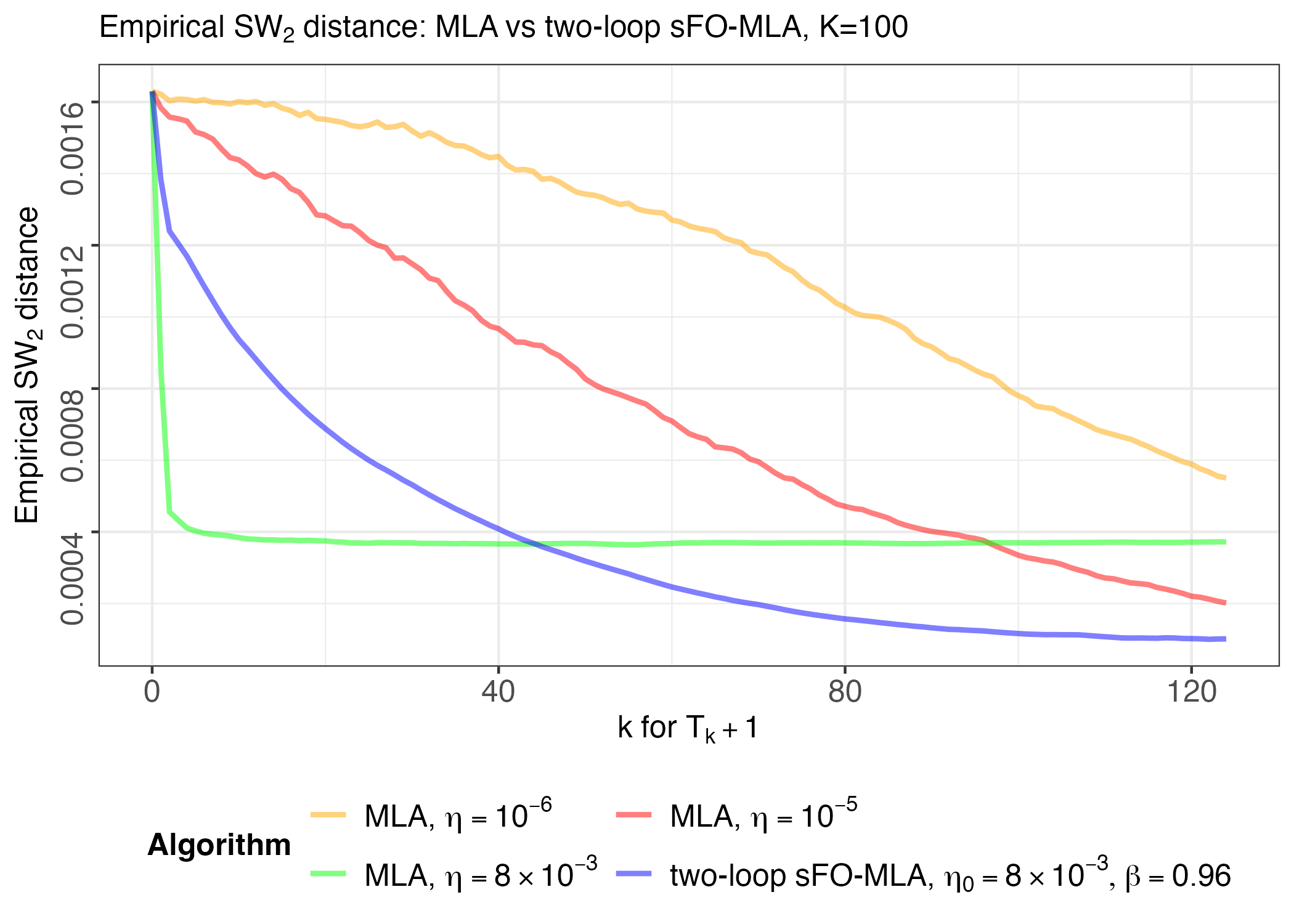}
    \end{minipage}
    \hfill
    \caption{Empirical order-2 sliced Wasserstein distance to MH-corrected MLA benchmark for $K=30,50,80,100$.}
    \label{fig:W2}
\end{figure*}

\begin{figure*}
    \centering
    \begin{minipage}[t]{0.23\textwidth}
        \centering
        \includegraphics[width=\linewidth, height = 4cm]{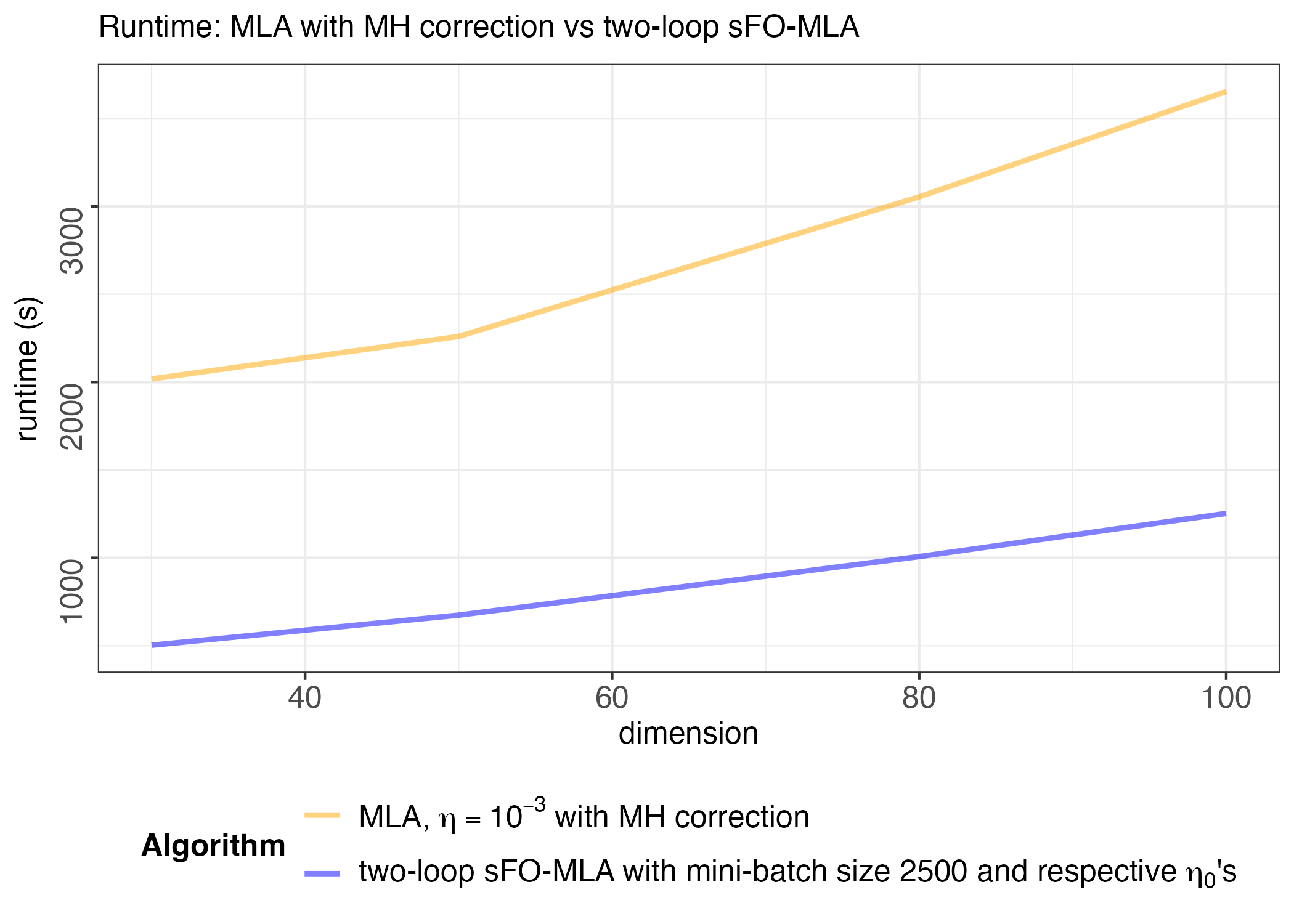}
    \end{minipage}
    \hfill
    \begin{minipage}[t]{0.23\textwidth}
        \centering
        \includegraphics[width=\linewidth, height = 4cm]{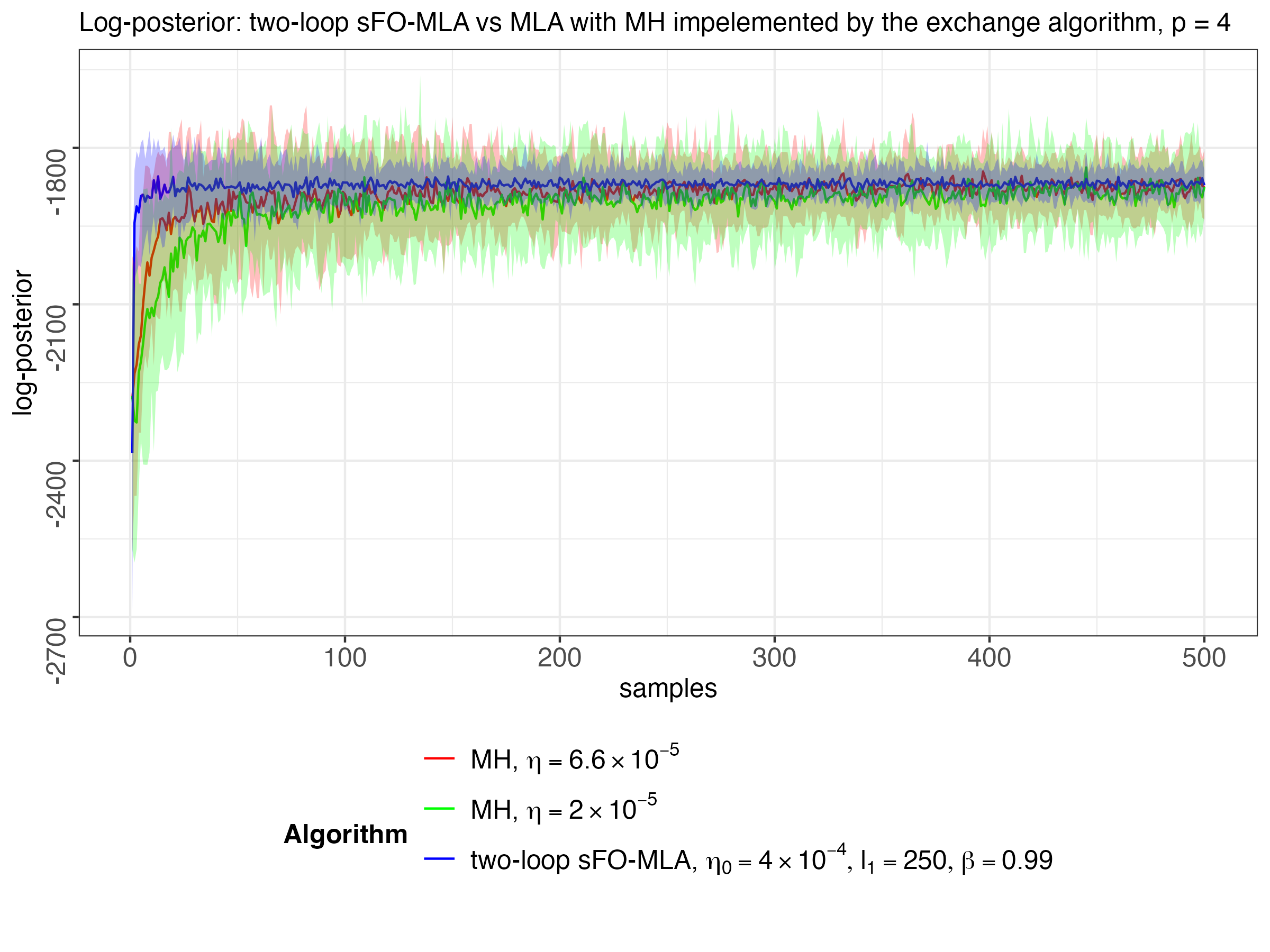}
    \end{minipage}
    \hfill
    \begin{minipage}[t]{0.23\textwidth}
        \centering
        \includegraphics[width=\linewidth, height = 4cm]{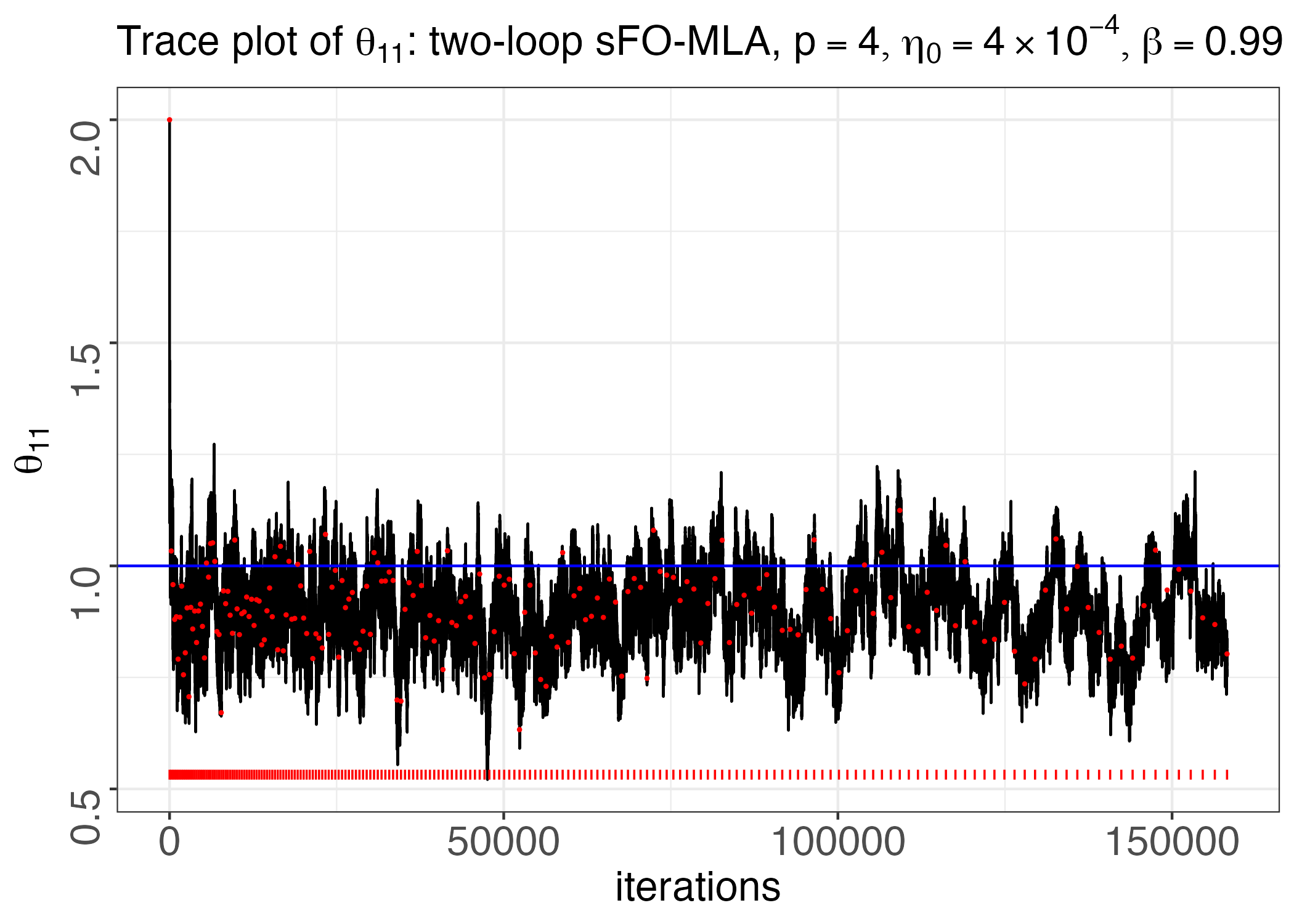}
    \end{minipage}
    \hfill
    \begin{minipage}[t]{0.23\textwidth}
        \centering
        \includegraphics[width=\linewidth, height = 4cm]{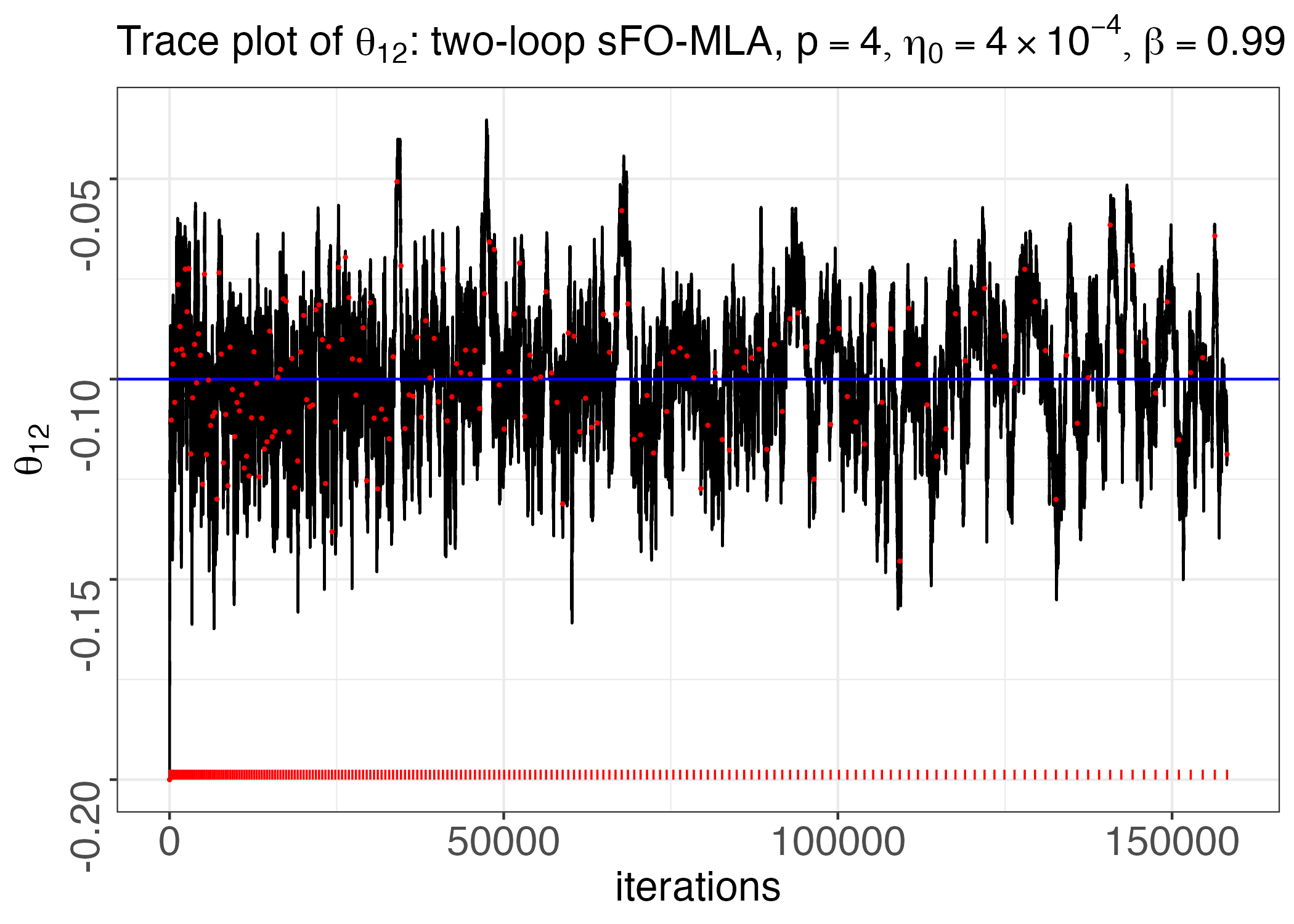}
    \end{minipage}
    \caption{Left to right: runtime of two-loop sFO-MLA and MH-corrected MLA in the simplex experiment; estimated log-posterior traces for two-loop sFO-MLA and the exchange algorithm in the PGM experiment with $p = 4$ and $90\%$ confidence bands; two-loop sFO-MLA traces of $\theta_{11}$ and $\theta_{12}$ for the PGM with $p = 4$.}
    \label{fig:traceplot}
\end{figure*}

%In our numerical experiments, we focus on two examples. The first one is a situation where the target distribution is supported over a simplex but computing the gradient involves large sums and the second one is an example which requires constrained sampling from a doubly intractable distribution. In both examples, we use the log-barrier mirror map; see Section \ref{sec:mirror_maps} for details. 
\subsection{Sampling over the simplex}\label{sec:simplex}
Consider observing  data $Z_i \overset{iid}{\sim} p_\theta(z) = \sum_{k=1}^K \theta_k q_k(z)$ where for simplicity we are going to assume that $q_k(\cdot)$ is completely known; we set this to be univariate Gaussian distribution with mean $k$ and standard deviation 0.2. The parameters of interest are the weights of the mixture distribution $\Theta = (\theta_1, \ldots, \theta_K)$. Of course, $\Theta \in \Delta_K = \{\omega \in \mathbb{R}^K: \omega_k \geq 0, \sum_{k=1}^K \omega_k = 1 \}$. Suppose for Bayesian inference, we specify a Dirichlet prior on $\Theta$, i.e. $\Theta \sim \text{Dirichlet}(\alpha_1, \ldots, \alpha_K)$. Then, the posterior distribution of $\Theta$ is $\pi(\Theta \mid Z_1, \ldots, Z_n) \propto \exp\{-f(\Theta)\} \mathbf{1}_{\Delta_K}(\Theta)$ where $f(\Theta) = -\sum_{i=1}^n \log \left( \sum_{k=1}^K \theta_k q_k(Z_i)\right) - \sum_{k=1}^K (\alpha_k - 1) \log \theta_k.$
We implement sampling in the $d=K-1$ free coordinates
$u=(\theta_1,\ldots,\theta_{K-1})$, with
$\theta_K=1-\sum_{r=1}^{K-1}u_r$; see Section~\ref{sec:mirror_maps} for further details. We compare two-loop sFO-MLA with a minibatch size of 25\% of the original data size, exact-gradient MLA, and MH-corrected MLA
with exact gradients~\citep{srinivasan2024fast}. The latter two methods are at an advantage here since both have access to the exact gradient whereas sFO-MLA uses an unbiased estimator of the gradient and decreases the step-sizes as prescribed by Lemma \ref{lem:fixedstepopt}.
%The latter accepts or
%rejects each proposal using~\eqref{eq:MH_ratio}. 
Experimental details
are in Section~\ref{sec:simplex_details}.

To compare two-loop sFO-MLA and fixed-step MLA, we measure the distance
between their empirical sample distributions and samples from MH-corrected
MLA used as a posterior benchmark. We use the \texttt{swdist} function
from the \texttt{R} package \texttt{T4transport}.
% AUTHOR CHECK: swdist computes a sliced Wasserstein distance. Confirm its
% order and settings, and update the W_2 labels in the supplied plot images.
The results are summarized in Figure~\ref{fig:W2}. Larger fixed step
sizes yield a faster initial decrease followed by a larger apparent error
floor, whereas smaller step sizes reduce this floor but converge more
slowly. The blue curves show that two-loop sFO-MLA combines rapid initial
progress with continued error reduction as the step size decreases.
The left panel of Figure~\ref{fig:traceplot} compares runtime with
MH-corrected MLA, which costs approximately two to three times as much in
these experiments.

\subsection{Doubly intractable models}\label{sec:PGM}
We consider sampling the posterior distribution of the Poisson Graphical model (PGM) \citep{yang2013poisson} which has an intractable normalizing constant. The probability function has the form $p(x;\theta) = E_\theta(x)/z(\theta)$. The parameter space $\Omega$ is defined as $\Omega = \{\theta^{p \times p}: z(\theta) < \infty, \theta= \theta'\}$. It is well-known that $\Omega = \{\theta: \theta_{jj} \in \mathbb{R}, \theta_{jk} = \theta_{kj} \leq 0 \text{ for } j<k\}$.
Given $X^{(i)} \overset{iid} {\sim}\text{PGM}(\theta)$, $ i = 1, \ldots, n$ and some prior $\pi(\theta)$, the objective is to sample from $\pi(\theta\mid X_1, \ldots, X_n)$ obtained using Bayes' theorem. To deal with the intractable term $\log z(\theta)$ in the posterior potential, we write $\nabla_\theta \log z(\theta) = \mathbb{E}_{p_\theta}[T(Y)]$, where $T(Y)$ has diagonal-coordinate entries $Y_j$ and off-diagonal-coordinate entries $2Y_jY_k$ ($j<k$), matching the parameterization in~\eqref{eq:pgm_mass}. Thus a stochastic oracle is constructed by considering sample averages of $T(Y)$ for observations generated from $p_\theta$. Writing $\theta^{p(p+1)/2} = \text{vech}(\theta)$ and $d = p(p+1)/2$, we put $p = 4$ and $5$ and sample $n = 500$ samples from the PGM with true $\theta^*$, a matrix with all diagonal entries equal to $1$ and all off-diagonal entries equal to $-0.1$ for each respective $p$. As an alternative to two-loop \ref{eq:FO}, we use the exchange algorithm \citep{murray2012mcmc}; details are provided in Section \ref{sec:PGM_details}. 
In the second through fourth panels of Figure \ref{fig:traceplot} the sample paths of the log-posterior averaged over 20 replicates and traces from one run for two coordinates of $\theta$, namely $\theta_{11}, \theta_{12}$ are shown for $p =4$. Corresponding results for $p = 5$ are reported in the Supplement.  
In particular, we see that by adjusting the step size after reaching the error floor for each inner loop, \ref{eq:FO} is able to begin with a large step size and mix fast. For the MH-based exchange algorithm only small step sizes result in  reasonable acceptance probabilities, but with slow mixing.

\section{Supplementary material}
The supplementary document contains proofs of all theorems and lemmas, further numerical experiment results, and results on maximum marginal likelihood estimation where latent variables are sampled using \ref{eq:FO}. Code and usage examples are available at \href{https://github.com/RuitingDeposit/sFO-MLA/}{https://github.com/RuitingDeposit/sFO-MLA/}.
%\clearpage
\bibliographystyle{plainnat}
\bibliography{ref_clean}
%%%%%%%%%%%%%%%%%%%%%%%%%%%%%%%%%%%%%%%%%%%%%%%%%%%%%%%%%%%%

\newpage
\onecolumn
\clearpage\pagebreak\newpage
\setcounter{secnumdepth}{3}
\setcounter{equation}{0}
\setcounter{page}{1}
\setcounter{table}{0}
\setcounter{section}{0}
\setcounter{subsection}{0}
\setcounter{figure}{0}
\renewcommand{\theequation}{S.\arabic{section}.\arabic{equation}}
\renewcommand{\thesection}{S.\arabic{section}}
\renewcommand{\thepage}{S.\arabic{page}}
\renewcommand{\thetable}{S.\arabic{table}}
\renewcommand{\thefigure}{S.\arabic{figure}}

\section*{\uppercase{Supplementary Material}}
\section{Proofs}\label{sec:supp_proofs}
\subsection{Proof of Theorem \ref{thm:MS-rec}}
\begin{proof}
Couple the exact-gradient and stochastic first-order chains using an
optimal coupling of their initial laws and common Gaussian
innovations. Define
\[
\Delta_t:=Y_t-\widehat Y_t,
\qquad
G_t:=g(Y_t)-g(\widehat Y_t),
\]
\[
\Xi_t
:=
\widehat g(\widehat Y_t,\zeta_{t+1})
-
g(\widehat Y_t),
\qquad
D_t
:=
A_\varphi(Y_t)-A_\varphi(\widehat Y_t).
\]
Then
\[
\Delta_{t+1}
=
\Delta_t-\eta G_t+\eta\Xi_t
+
\sqrt{2\eta}\,D_tN_{t+1}.
\]
Conditioning on $\mathcal F_t$, using $\E\!\left[ \widehat g(\widehat Y_t,\zeta_{t+1}) \,\middle|\, \mathcal F_t \right] = g(\widehat Y_t),$ and
conditional independence of the oracle and Gaussian innovations,
gives
\begin{align}
\E\!\left[
\|\Delta_{t+1}\|^2
\,\middle|\,
\mathcal F_t
\right]
={}&
\|\Delta_t-\eta G_t\|^2
+
\eta^2
\E\!\left[
\|\Xi_t\|^2
\,\middle|\,
\mathcal F_t
\right]
+
2\eta\|D_t\|_{\mathrm F}^2.
\end{align}
By Assumptions \ref{ass:A1} - \ref{ass:A4},
\begin{align}
\E\!\left[
\|\Delta_{t+1}\|^2
\,\middle|\,
\mathcal F_t
\right]
&\le
\left(
1-2m\eta+M^2\eta^2+2\alpha\eta
\right)
\|\Delta_t\|^2
+
\sigma^2d\eta^2
\nonumber\\
&=
\left(
1-2\gamma\eta+M^2\eta^2
\right)
\|\Delta_t\|^2
+
\sigma^2d\eta^2.
\end{align}
Therefore, if $\eta\le\gamma/M^2$,
\[
\E\!\left[
\|\Delta_{t+1}\|^2
\,\middle|\,
\mathcal F_t
\right]
\le
(1-\gamma\eta)\|\Delta_t\|^2
+
\sigma^2d\eta^2.
\]
Iteration yields
\[
\E\|\Delta_t\|^2
\le
(1-\gamma\eta)^t\E\|\Delta_0\|^2
+
\frac{\sigma^2d}{\gamma}\eta.
\]
Consequently,
\begin{align}
\Wtwo\!\left(
\widehat\pi^*_{t,\eta,\varphi},
\pi^*_{t,\eta,\varphi}
\right)
\le
e^{-\frac{\gamma}{2}\eta t}
\Wtwo\!\left(
\widehat\pi^*_{0,\eta,\varphi},
\pi^*_{0,\eta,\varphi}
\right)
+
\sigma\sqrt{\frac{d\eta}{\gamma}}.
\label{eq:FO-perturbation}
\end{align}
Since the two chains have the same initial law, the transient term
in~\eqref{eq:FO-perturbation} vanishes. Hence
\[
\Wtwo\!\left(
\widehat\pi^*_{t,\eta,\varphi},
\pi^*_{t,\eta,\varphi}
\right)
\le
\sigma\sqrt{\frac{d\eta}{\gamma}}.
\]

The triangle inequality now gives
\[
\Wtwo\!\left(
\widehat\pi^*_{t,\eta,\varphi},
\pi_\varphi^*
\right)
\le
\Wtwo\!\left(
\widehat\pi^*_{t,\eta,\varphi},
\pi^*_{t,\eta,\varphi}
\right)
+
\Wtwo\!\left(
\pi^*_{t,\eta,\varphi},
\pi_\varphi^*
\right).
\]
Applying the exact-gradient MLA bound used in \cite{li2022mirror} to the second term, and using
\[
\pi^*_{0,\eta,\varphi}
=
\widehat\pi^*_{0,\eta,\varphi},
\]
proves~\eqref{eq:main-W2-generic-fo}.
\end{proof}
\subsection{Proof of Lemma \ref{lem:fixedstepopt}}
For the first assertion, the definition of $T^*(\eta)$ ensures that, for every $t\ge T^*(\eta)$, we have $A_0e^{-\gamma\eta t} \le B_0\eta^{1/2}.$
Substitution into~\eqref{gaja} gives
$w(t,\eta) \le 2B_0 \eta^{1/2}.$ For the second assertion, fix $T$ and differentiate:
$\frac{\partial}{\partial\eta}w(T,\eta)
=
-\gamma T A_0e^{-\gamma T\eta}
+
(1/2) B_0\eta^{-1/2}.$
Thus any interior stationary point satisfies
$\gamma T A_0e^{-\gamma T\eta}
=
(1/2) B_0\eta^{-1/2}.$
Setting $c:=\gamma T$
and rearranging gives $$\eta e^{-2c\eta} = \dfrac{B_0^2}{4c^2 A_0^2} \iff (-2c\eta) e^{-2c\eta} = -\dfrac{B_0^2}{2cA_0^2}. $$
Using the Lambert $W$ function \citep{2020ebe} we thus obtain
$$\eta = -\dfrac{1}{2\gamma T} W_k\left( -\dfrac{B_0^2}{2\gamma T A_0^2}\right), \quad k\in \{0, -1\}.$$
For all sufficiently large \(T\), the argument of the Lambert $W$ function belongs to $(-e^{-1},0)$ , yielding two real stationary points. The smaller stationary point, corresponding to the principal branch $W_0$, is a local maximum, whereas the larger stationary point,
$$\eta^*(T) = -\frac{1}{2\gamma T}W_{-1}\left(-\frac{B_0^2}{2\gamma T A_0^2}\right),$$
is a local minimum. Moreover, $\eta^*(T)=O((\log T)/T)\to0$, so for any fixed $\bar\eta>0$, $\eta^*(T)\in(0,\bar\eta]$  for all sufficiently large $T$, and $\eta^*(T)$  is eventually the global minimizer on $(0,\bar\eta]$.
%For all sufficiently large $T$, the argument on the right belongs to $(-e^{-1},0)$, and the larger stationary point, corresponding
%to the $W_{-1}$ branch, lies in $(0,\bar\eta]$ and is the minimizer.
This yields the stated expression for $\eta^*(T)$. Finally, using the standard expansion of $W_{-1}(x)$ as
$x\uparrow0$ gives $\eta^*(T)
=
\frac{1}{2\gamma}\frac{\log T}{T}
+
O\!\left(
\frac{\log\log T}{T}
\right).$
Substitution into~\eqref{gaja}, together with the first-order
condition above, shows that the transient term is of smaller order
than $B_0(\eta^*(T))^{1/2}$, and hence $w\!\left(
T,\eta^*(T)
\right)
\sim
B_0
\left(
\frac{1}{2\gamma}
\right)^{1/2}
\left(
\frac{\log T}{T}
\right)^{1/2}.$

\subsection{Proof of Theorem \ref{thm:two-loop-short}}
Since $\eta_k\ell_k\geq c_0$,~\eqref{eq:epoch-generic} gives $D_{k+1}
\leq
aD_k+B_0\eta_0^{1/2}\beta^{k/2}, a:=e^{-\gamma c_0}.$
Iterating over $k$ yields $D_k
\leq
a^kD_0+
B_0\eta_0^{1/2}
\sum_{j=0}^{k-1}a^{k-1-j}\beta^{j/2}.$
Equivalently, $D_k
\leq
a^kD_0+
B_0\eta_0^{1/2}\beta^{(k-1)/2}
\sum_{j=0}^{k-1}q_{1/2}^{\,j}, q_{1/2}=a\beta^{-1/2}.$
The three cases in~\eqref{eq:two-loop-rates} now follow from the
geometric sum and the fact that $T_k\asymp\beta^{-k}$.
\subsection{Proof of Theorem \ref{thm:twoloop-poly}}
By Assumptions~\ref{ass:A5} and~\ref{ass:fixedstep-bias},
the triangle inequality gives, for any $\mu\in\mathcal P_2(\mathbb R^d)$, $W_2(\mu P_\eta^t,\pi^*)
\leq
C_{\rm mix}(1+\eta t)^{-q}W_2(\mu,\pi^*)
+
B_r\eta^r,$
where $B_r=(1+C_{\rm mix})C_{\rm bias}$.
Applying this bound over epoch $k$ and using
$\eta_k\ell_k\geq c_0$ yields $D_{k+1}\leq \rho D_k+B_r\eta_0^r\beta^{rk}, \rho=C_{\rm mix}(1+c_0)^{-q}.$
Iteration gives $D_k
\leq
\rho^kD_0+
B_r\eta_0^r
\beta^{r(k-1)}
\sum_{j=0}^{k-1}Q_r^j,Q_r=\rho\beta^{-r}.$
The three regimes follow from the geometric sum and
$T_k\asymp\beta^{-k}$.
\section{Details of numerical experiments}
\subsection{Mirror maps}\label{sec:mirror_maps}
In both examples of Section \ref{sec:numerical_experiments}, the constrained space is defined by linear boundaries. For the mixture model, we have $\theta_i \geq 0$ for $i = 1,2, \dots, K$ and $\sum_{i=1}^K \theta_i = 1$. Since only the first $K-1$ mixture coefficients are free, we can reformulate the constraints as $\theta_i \geq 0$ for $i = 1,2,\dots, K-1$ and $\sum_{i = 1}^{K-1} \theta_i \leq 1$. It can be summarized as 
$$\begin{pmatrix}
    -I_{K-1} \\
    \mathbf{1}^\top_{K-1}
\end{pmatrix} \theta \leq \begin{pmatrix}
    \mathbf{0}_{K-1}\\
    1
\end{pmatrix},$$
where $I_{K-1}$ is the $(K-1)$-dimensional identity matrix, $\mathbf{1}_{K-1}$ is a column vector of length $K-1$ containing just ones, and similarly $\mathbf{0}_{K-1}$ contains just zeros. 

For PGM, the parameter $\theta$ matrix is a symmetric matrix, the parameter is $\mathrm{vech}(\theta) = (\theta_{11}, \theta_{12}, \dots, \theta_{1p}, \theta_{22}, \dots, \theta_{pp}) \in \mathbb{R}^{p(p+1)/2}$, where the indices are arranged in the lexicographic order. Clearly $\mathrm{vech}(\mathbf{1}_p\mathbf{1}_p^\top - I_p) \in \mathbb{R}^{p(p+1)/2}$ is a vector containing zeros and ones such that ones only appear where the index corresponds to a non-diagonal entry of $\theta$. Hence, the constraints can be expressed as $\mathrm{diag}(\mathrm{vech}(\mathbf{1}_p \mathbf{1}_p^\top - I_p)) \mathrm{vech}(\theta) \leq \mathrm{vech}(I_p)$. Therefore, if we denote $\mathrm{vech}(\theta)$ by $\theta$ by an abuse of notations, in both examples we have a constraint of the form $A\theta \leq b$, where $A,b$ are constants that do not depend on $\theta$.

For both scenarios, we have chosen the mirror map to be a regularized log barrier map 
\begin{equation}\label{eq:log_barrier}
\varphi(x) = - \sum_{j=1}^{m_A} \log(b_j-a_j^\top x) + \frac{\lambda}{2} \|x\|_2^2,
\end{equation}
where $a_j^\top$ is row $j$ of $A$, $m_A$ is its number of rows, $A$ and $b$ have been discussed above, and $\lambda$ is a regularizer fixed at $1$ in all experiments. The regularizer is added to ensure that $\nabla^2 \varphi$ is always well-behaved. It follows that $$\nabla \varphi(x) =  -A^\top \frac{1}{Ax - b} + \lambda x, $$ 
in which $1/(Ax - b)$ is a vector obtained by the obvious element-wise division. For both examples of Section \ref{sec:numerical_experiments}, $\nabla \varphi$ is a bijection between the interior of the constrained space and $\mathbb{R}^d$, where $d$ denotes the respective dimension of the problems. It is equal to $K-1$ for the mixture model, and $p(p+1)/2$ for the PGM. 

Next, we present how we compute the inverse $(\nabla \varphi)^{-1} = \nabla \varphi^*$. In the mixture model, the equation $\nabla\varphi(x) = y$ expands into 
\begin{align*}
&\nabla\varphi(x) = y \\
&\implies x_i - \frac{1}{x_i} + \frac{1}{1-\sum_{j=1}^{K-1}x_j} = y_i, \quad i = 1,2,\dots, K-1
\end{align*}
The equations don't have a closed form solution, so we solve it numerically. Set $t = \frac{1}{1 - \sum_{j= 1}^{K-1}x_j} $. Treating $t$ as a constant, we can solve the equation $$x_i - \frac{1}{x_i} = y_i - t$$ 
as it translates to $x_i^2 - (y_i - t)x_i - 1 = 0$. Since each $x_i \geq 0$ due to the constraint, the solution is 
$$x_i(t) = \frac{y_i - t + \sqrt{(y_i - t)^2 + 4}}{2},$$ 
and then by the definition of $t$, we must have $1-\sum_{j=1}^{K-1} x_i(t) = \frac{1}{t}$. Let $F(t) = \sum_{j=1}^{K-1} x_i(t) + \frac{1}{t} - 1$. Observe that $t > 0$ in the interior of our constraint, and 
$$x_i'(t) = -\frac{1}{2}\left(1- \frac{t-y_i}{\sqrt{(t-y_i)^2 + 4}}\right) < 0,$$ 
so $F(t)$ is a strictly decreasing function on $(0,\infty)$. In addition, $\lim_{t\to+\infty} F(t) = -1$ and $\lim_{t\to0+}  F(t) = +\infty$, so we are able to find a unique solution to $F(t) = 0$ using the bisection method, which in turn gives us the solution $x_i(t)$'s.

In PGM, let $\nabla \varphi (x) = y$ for an arbitrary $y \in \mathbb{R}^{p(p+1)/2}$. Using $\{(i,j): 1\leq i \leq j \leq p\}$ in the lexicographic order to index the entries of this vector, we have $x_{ii} = y_{ii}$ for $i = 1,2,\dots, p$, and 
$$-\frac{1}{x_{ij}} + x_{ij} = y_{ij}, \quad i \neq j.$$  
So we only need to solve for $x_{ij}$ when $i\neq j$. The $x_{ij}$'s satisfy the quadratic equation $x_{ij}^2 - y_{ij}x_{ij} - 1 = 0$, which has solutions $\frac{y_{ij} \pm \sqrt{y_{ij}^2 + 4}}{2}$, and the constraint $x_{ij} \leq 0$ tells us it must equal $\frac{y_{ij} - \sqrt{y_{ij}^2 + 4}}{2}$. 

%\subsection{Additional details}
\subsection{Experimental details of Section \ref{sec:simplex}}\label{sec:simplex_details}

\begin{figure*}[htb]
    \centering
    \begin{minipage}[t]{0.23\textwidth}
        \centering
        \includegraphics[width=\linewidth, height = 4cm]{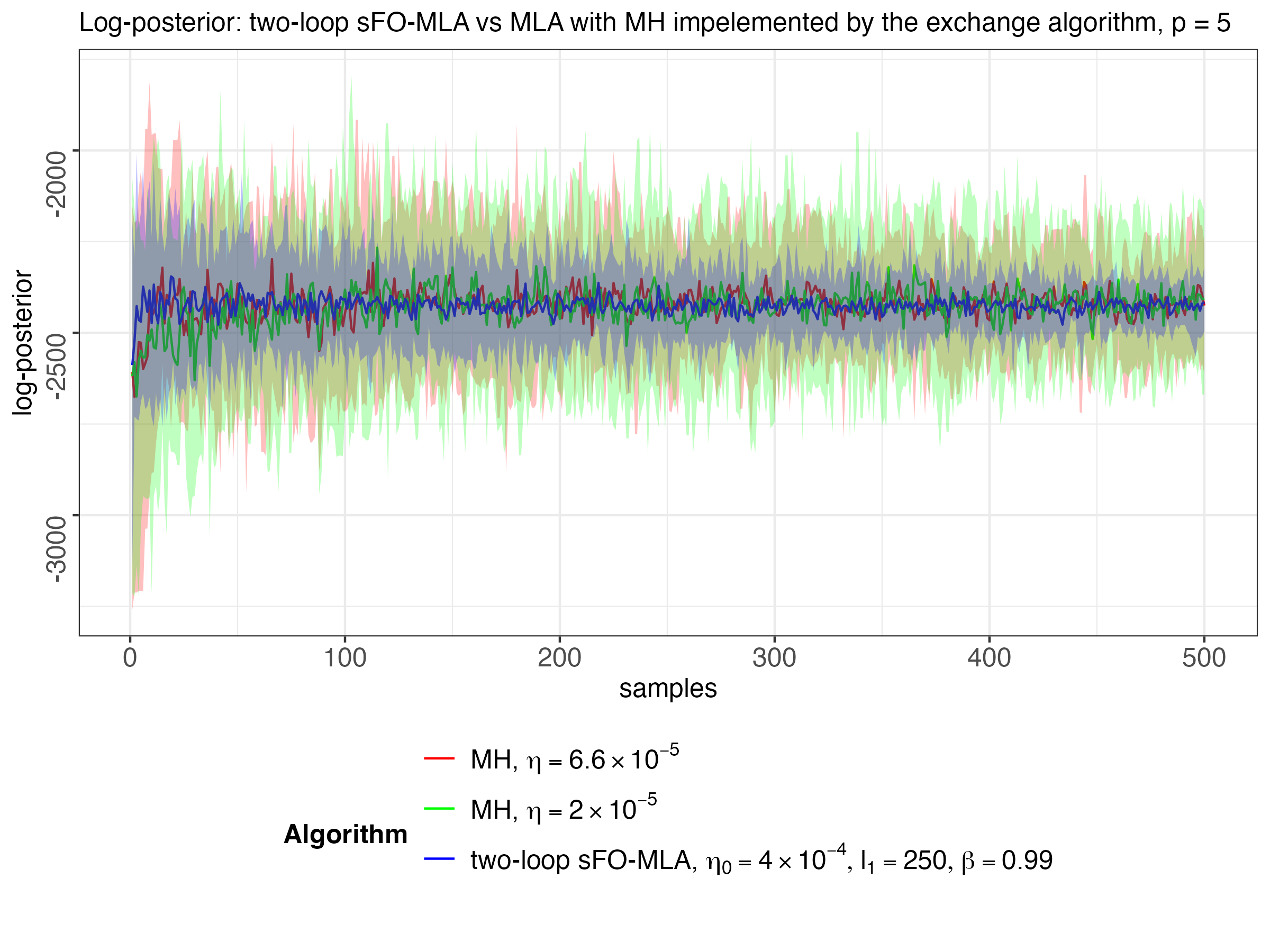}
    \end{minipage}
    \begin{minipage}[t]{0.23\textwidth}
        \centering
        \includegraphics[width=\linewidth, height = 4cm]{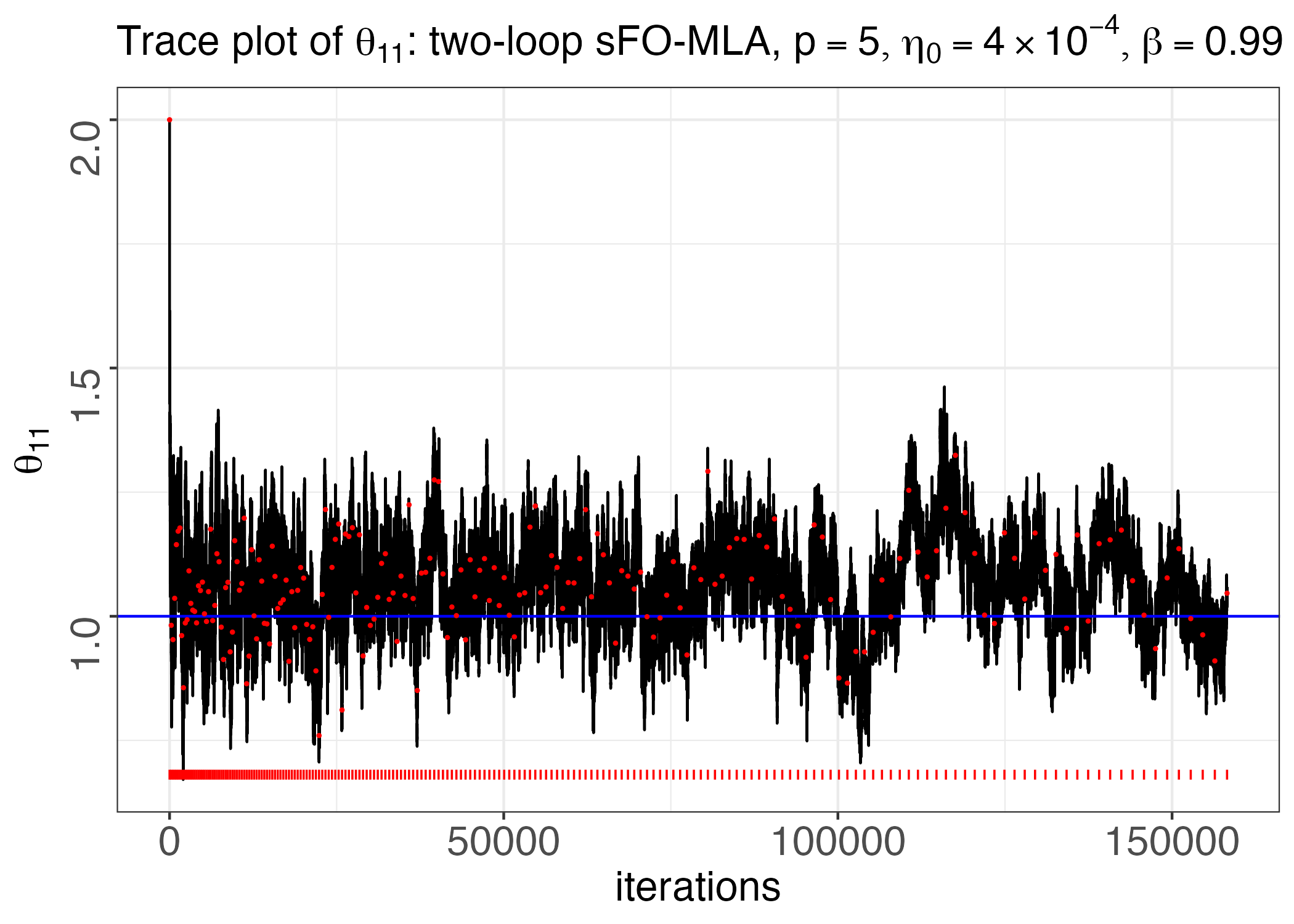}
    \end{minipage}
    \begin{minipage}[t]{0.23\textwidth}
        \centering
        \includegraphics[width=\linewidth, height = 4cm]{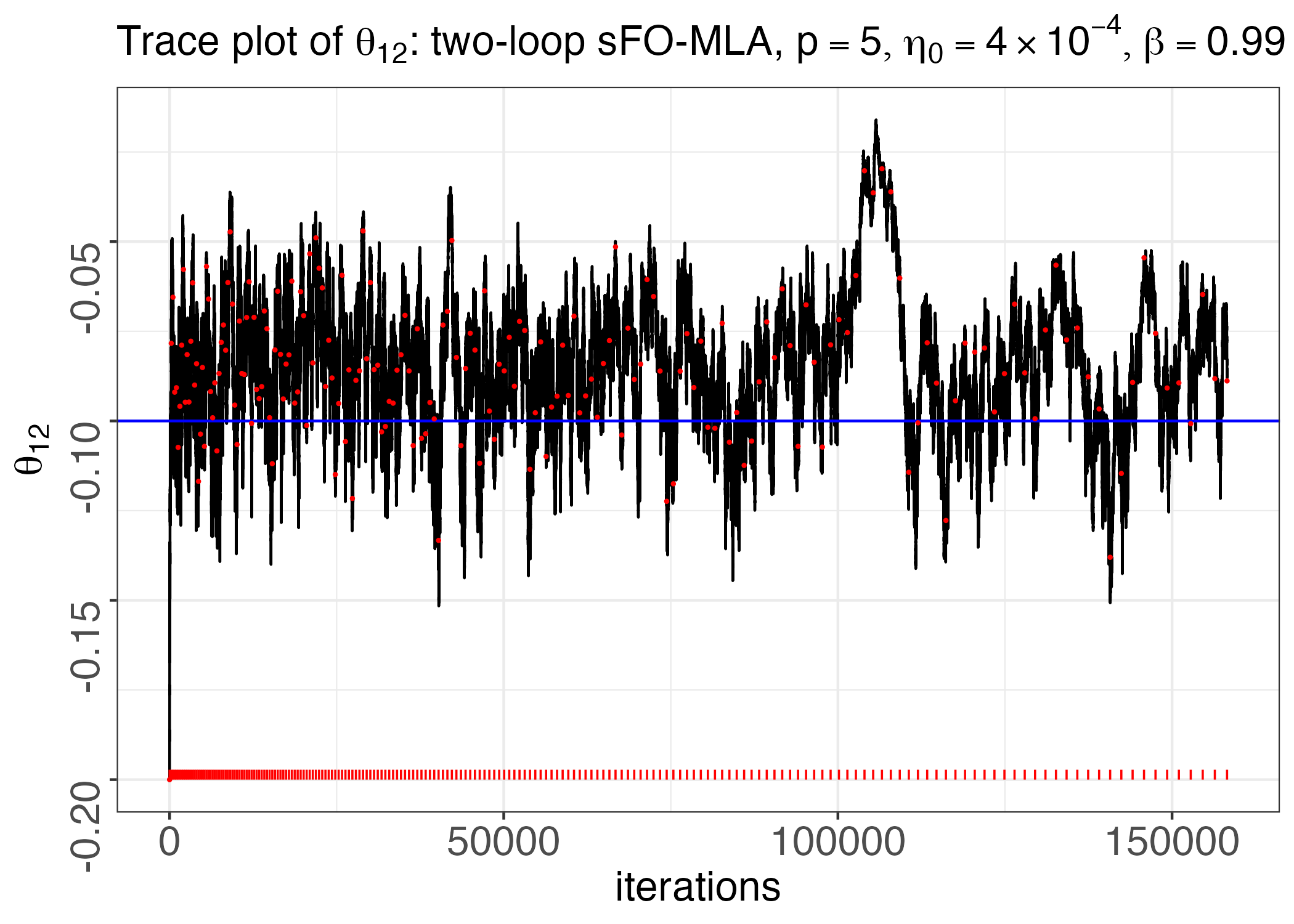}
    \end{minipage}
    \caption{Left to right: estimated log-posterior traces for two-loop sFO-MLA and the exchange algorithm in the PGM experiment with $p = 5$ and $90\%$ confidence bands; two-loop sFO-MLA traces of $\theta_{11}$ and $\theta_{12}$ for the PGM with $p = 5$.}
    \label{fig:supplement_figure}
\end{figure*}

Write $\widetilde f(u)=f(\Theta(u))$ and
$h_i(u)=\sum_{r=1}^K\theta_r(u)q_r(Z_i)$.
The reduced-coordinate gradient is, for $k=1,\ldots,K-1$,
\[
\partial_{u_k}\widetilde f(u)
=-\sum_{i=1}^n\frac{q_k(Z_i)-q_K(Z_i)}{h_i(u)}
-\frac{\alpha_k-1}{u_k}+\frac{\alpha_K-1}{\theta_K}.
\]
When $n$ is large, choose a uniform minibatch $B$ of size $m$ from
$\{1,\ldots,n\}$. An unbiased stochastic gradient is
\[
\widehat{\partial_{u_k}\widetilde f}(u)
=-\frac{n}{m}\sum_{i\in B}\frac{q_k(Z_i)-q_K(Z_i)}{h_i(u)}
-\frac{\alpha_k-1}{u_k}+\frac{\alpha_K-1}{\theta_K}.
\]
This estimate is used in the updates of two-loop sFO-MLA~\eqref{eq:FO}.
The observed data were generated by setting $\theta_1 = \ldots = \theta_4 = 1/(2K)$, $\theta_{K-3} = \ldots = \theta_{K} = 3/(2K)$ and the rest were set to $1/K$. All Dirichlet prior parameters $\alpha_1 = \ldots = \alpha_K = 2$. We consider $K \in \{30, 50, 80, 100\}$ and for each $K$, we sample $n = 10000$ independent $Z_i$'s from the distribution $p_\theta(z)$. In the two-loop sFO-MLA, for all $K$'s considered, we set the minibatch size $m = 2500$, the geometric decay parameter $\beta = 0.96$, the initial inner loop size $\ell_0 = 15$, the initial sample $(\frac{1}{K}, \frac{1}{K}, \dots, \frac{1}{K})$, and we let the algorithm run for $M = 135$ outer loop samples, including the initial sample, producing a total of $T_{M-1} + 1= 85203$ samples. \ref{eq:mla-short} and MH corrected MLA were also run for $T_{M-1} + 1$ iterations. Different initial values of $\eta_0$ are taken for different $K$'s. They are $3\times 10^{-3}$, $5\times 10^{-3}$, $7\times 10^{-3}$, and $8\times 10^{-3}$ for $K=30,50,80,100$, respectively. The three selected step-sizes for \ref{eq:mla-short} are $ 10^{-6}, 10^{-5}$, and $\eta_0$, which is defined differently above for each dimension; for the MH corrected version we let $\eta = 10^{-3}$.
For each $K$, we start at $(\frac{1}{K},\frac{1}{K},\dots, \frac{1}{K})$.

% AUTHOR CHECK: add PGM prior, initialization, oracle sampling method,
% step-size/epoch settings, runtime hardware, and Monte Carlo settings.
\subsection{Experimental details of Section \ref{sec:PGM}}\label{sec:PGM_details}
Suppose $X = (X_1, \ldots, X_p)$ is a random vector such that each $X_j$ has support $\mathcal{S} = \{0, 1, 2, \ldots\}$. Then given a graph $G = (V, E)$, where $V$ is the set of $p$ nodes and $E$ is the set of edges, the PGM has mass function
\begin{equation}\label{eq:pgm_mass}
    p(\mathbf{x};\theta) = \frac{1}{z(\theta)} \exp\left\lbrace \sum_{j \in V} \theta_j \mathbf{x}_j - \log (\mathbf{x}_j !) + \sum_{(j,k) \in E} \theta_{jk} \mathbf{x}_j \mathbf{x}_k\right\rbrace = \frac{E_\theta(\mathbf{x})}{z(\theta)}.
\end{equation}
Consider independent and identically distributed observations $X^{(i)} \sim \text{PGM}(\theta)$, $ i = 1, \ldots, n$ and for Bayesian inference a prior $\pi(\theta)$ is elicited on $\Omega$. Then the posterior distribution is 
\begin{align}\label{eq:PGM_posterior}
    \pi(\theta \mid X^{(1)}, \ldots, X^{(n)}) &\propto \left[\prod_{i=1}^n  \frac{E_\theta(\mathbf{x}^{(i)})}{z(\theta)}\right] \pi(\theta)  \nonumber \\
    &\propto e^{-f(\theta)}, \quad \theta \in \Omega
\end{align}
where $f(\theta) = -\sum_{i=1}^n \log E_\theta(\mathbf{x}^{(i)}) + n \log z(\theta) - \log \pi(\theta)$. In a traditional MCMC algorithm, implementing the MH correction for this posterior distribution requires evaluating the potential at the proposed state and the current state, both of which involve the intractable normalizing constant $z(\theta)$.

The exchange algorithm proposed by \cite{murray2012mcmc} targets posteriors of this form. Suppose one uses a mirror map $\varphi(\theta)$ to map the parameter space onto the unconstrained space $\mathbb{R}^{p(p+1)/2}$ and forms stochastic proposals using the same unbiased estimator of $\nabla_\theta \log z(\theta)$ as in Section \ref{sec:PGM}. Then at each time step, a candidate $\theta$ is proposed which has density as in Remark \ref{rem:proposal_tractability}.
The exchange algorithm defines a transition kernel for $(\theta, \hat{g})$ where $\hat{g} = \hat{g}(\nabla \varphi(\theta))$. Define the state of the chain at time $t$ to be $\theta_t$ and the stochastic estimate of the gradient is $\hat{g}(\nabla \varphi(\theta_t))$. 
This algorithm uses the same proposal mechanism at every iteration as defined in \eqref{eq:FO} but this new proposal (say $(\theta^\star, \hat{g}(\nabla \varphi(\theta^\star)))$) is accepted with probability 
$$\alpha\{(\theta_t, \hat{g}), ( \theta^\star, \hat{g}^\star )\} = \min\left\lbrace 1, \dfrac{\pi(\theta^\star\mid X^{(1)}, \ldots, X^{(n)})\hat{q}(\theta_{t}\mid \theta^\star, \hat{g}(\nabla \varphi(\theta^\star)))}{\pi(\theta_{t}\mid X^{(1)}, \ldots, X^{(n)}) \hat{q}(\theta^\star \mid \theta_{t}, \hat{g}(\nabla \varphi(\theta_t)))}\right\rbrace,$$
where $\hat{q}(\cdot\mid \theta_t, \hat{g}(\nabla \varphi(\theta_t)))$ is the conditional distribution of proposed state given the current state $\theta_t$ and a stochastic estimate $\hat{g}(\nabla \varphi(\theta_t)) $ of the gradient. 
%It is related to the proposal distribution \eqref{eq:mld_proposal} by the relation $q(\theta^\star \mid \theta_t) = \int \hat{q}(\theta^\star \mid \theta_t, \hat{g}(\nabla \phi(\theta_t))) p(\hat{g}(\theta_t) \mid \theta_t) d\hat{g}$ so that it is the Gaussian density multiplied by the Jacobian term in the integral in \eqref{eq:mld_proposal}. Also, $\pi(\theta\mid X^{(1)}, \ldots, X^{(n)})$ is the posterior distribution given in \eqref{eq:PGM_posterior}. 
Since the acceptance probability involves the intractable $z(\theta_t)^{n}/z(\theta^\star)^n$, we instead compute the exchange ratio \cite{murray2012mcmc}
\begin{equation}\label{eq:exchange_ratio}
    \alpha_{EX}\{(\theta_t, \hat{g}), ( \theta^\star,  \hat{g}^\star)\} = \min\left\lbrace 1, \left[\prod_{i=1}^n\dfrac{E_{\theta^\star}(X^{(i)})E_{\theta_t}(Y^{(i)})}{E_{\theta_{t}}(X^{(i)})E_{\theta^\star}(Y^{(i)})}\right] \dfrac{\pi(\theta^\star)}{ \pi(\theta_{t})} \dfrac{q(\theta_t\mid \theta^\star, \hat{g}^\star)}{q(\theta^\star\mid \theta_t, \hat{g})}\right\rbrace, 
\end{equation}
with $Y^{(i)} \overset{iid}{\sim} p_{\theta^\star}(\cdot)$.
% AUTHOR CHECK: exact invariance requires exact auxiliary draws and a
% specified augmented proposal for the retained stochastic gradient.
 The resulting Markov chain with this Metropolis-Hastings correction is a chain defined on an augmented space such that the marginal distribution over $\theta$ is the target posterior distribution $\pi(\theta \mid X^{(1)}, \ldots, X^{(n)})$. In practice, we cannot sample $Y^{(i)}$'s exactly from $p_{\theta^*}(\cdot)$, so in the experiments we use a Gibbs sampler that samples iteratively from each node-conditional distribution; we use 200 Gibbs iterations. We note here that ideally one could use the perfect sampling idea from \cite{propp1996exact} but scaling this to high-dimensions is challenging and has been reported elsewhere in the literature \citep{chen2024likelihood}.
 Therefore, the exchange algorithm is also approximated, and this again demonstrates the difficulty of making the MH step work for PGM. As for the log-posterior in Figure~\ref{fig:traceplot}
 and \ref{fig:supplement_figure}, they are also estimated because $z(\theta)$ is intractable. At each iteration, we estimate $z(\theta)$ using 5000 importance samples $Y_1, \ldots, Y_{5000}$ from $p_{\mathrm{diag}(\theta)}$, and then we compute the log-posterior by the estimate $\log \hat{z}(\theta)$ where $$\hat{z}(\theta) = z(\text{diag}(\theta))\frac{1}{5000} \sum_{i=1}^{5000} \dfrac{E_\theta(Y_i)}{E_{\mathrm{\text{diag}(\theta)}}(Y_i)},$$ 
 where $z(\text{diag}(\theta))$ is analytically available since it is the normalizing constant of the distribution of $p$ independent Poisson random variables.
 This is an unbiased estimator of $z(\theta)$; see \cite{geyer1991markov}. 
 The log-posterior traces in the plots are obtained by averaging 20 such traces.  %Indeed, using elementary calculations it can be shown that this chain has $\pi(\theta \mid X^{(1), \ldots, X^{(n)}}) p(\hat{g}(\nabla \varphi(\theta)) \mid \theta)$ as its stationary distribution. This algorithm can be implemented in both the zeroth and first order oracle scenarios. In the zeroth order case, $\hat{g}$ is either the forward-difference or centered-difference estimators of \eqref{eq:FDCD}, and in the first-order case $\hat{g}$ is simply $\nabla F$.

\section{Example: Maximum Marginal Likelihood Estimation}\label{sec:mmle}
The practical usefulness of Lemma \ref{lem:fixedstepopt} is clearly illustrated by considering the problem of maximum marginal likelihood estimation. Suppose we observe independent and identically distributed (iid) data $Y_i \sim p_\theta, \, i= 1, \ldots, n$ where $p_\theta(y) = \int_{z} f_\theta(y, z) dz$. Here, $z$ is a latent variable. This situation arises frequently in statistics and machine learning, e.g. hidden Markov models, missing data problems, Boltzmann machines etc. are among a few. To obtain the maximum likelihood estimate in this case, we iterate under standard regularity conditions
\begin{equation}\label{eq:MMLE}
    \theta_{t+1} = \theta_t + \gamma_t \sum_{i=1}^n \nabla_\theta \log p_{\theta_t}(Y_i),
\end{equation}
where $\gamma_t>0$ is a sequence of step-sizes. Suppose we are interested in the case when $p_\theta$ and (hence) $\nabla \log p_\theta$ do not admit a tractable form, but the joint distribution of the $(y,z)$ is tractable. Under such a setup, we have from Fisher's identity~\citep{cappe2005inference} that $\nabla \log p_\theta(Y) = \E_{p_\theta(z\mid y)}[\nabla \log f_\theta(y, z)]$ and replace the total observed-data score in \eqref{eq:MMLE} by $\E_{p_\theta(\mathbf{Z}\mid \mathbf{Y})}[\nabla \log f_\theta(\mathbf{Y}, \mathbf{Z})]$. Here, $\mathbf{Y} = (Y_1, \ldots, Y_n)$ and $\mathbf{Z} = (Z_1, \ldots, Z_n)$, and we use this notation to indicate that the latent variables may or may not be conditionally independent given the observed data. Although Fisher's identity provides an attractive alternative, in most practical applications, either $p_\theta(\mathbf{Z} \mid \mathbf{Y})$ is not standard or the expectation of the score of the complete data likelihood is not available or both. In such cases, it is standard to consider a Monte Carlo approximation of $\E_{p_\theta(\mathbf{Z}\mid \mathbf{Y})}[\nabla \log f_\theta(\mathbf{Y}, \mathbf{Z})]$. Writing $S^{\mathbf{Y}}_\theta(\mathbf Z)=\nabla_\theta\log f_\theta(\mathbf Y,\mathbf Z)$ for the total complete-data score, one then has the (stochastic) iterates
\begin{equation}\label{eq:sMMLE}
    \theta_{t+1} = \theta_t + \gamma_t S^{\mathbf{Y}}_{\theta_{t}}(\mathbf{Z}_{t+1}), \quad \mathbf{Z}_{t+1} \sim p_{\theta_t}(\mathbf{Z} \mid \mathbf{Y}),
\end{equation}
which, under suitable regularity and stability conditions, converges to a stationary point of the observed likelihood when the step-sizes satisfy $\sum \gamma_t = \infty$ and $\sum \gamma_t^2 < \infty$ \citep{robbins1951stochastic}. This guarantee critically assumes our ability to generate $\mathbf Z_{t+1} \sim p_{\theta_t}(\cdot \mid \mathbf Y)$ exactly which is rarely the case. In fact, often this distribution is constrained and $-\log p_\theta(z \mid y)$ is available only through a first-order oracle. Classical examples are multivariate probit models with random effects, probit ordinal models, hierarchical models etc. Next, we imagine a scenario where sFO-MLA was run with step-size $\eta_t$ to generate $Z_{t+1}$ and hence $S_{\theta_t}^{\mathbf Y}(\mathbf Z_{t+1})$ is no longer an unbiased estimate of the complete-data score. The following theorem identifies an exact relation between the optimization step-sizes $\gamma_t$, sFO-MLA step-sizes $\eta_t$ to attain convergence of $\theta_t$. For simplicity let us assume that $\theta^\star = \argmax_\theta \ell(\theta) = \argmax_\theta \sum_{i=1}^n \log p_\theta(Y_i) $ exists and is unique.
\begin{theorem}\label{thm:mmle_convergence}
    Suppose there exists $m>0$ such that $\langle \theta - \theta^\star, \nabla \ell(\theta) \rangle \leq - m \norm{\theta - \theta^\star}^2$ for all $\theta$ and the complete-data score $S_\theta^{\mathbf Y}(z) = \nabla_\theta \log f_\theta(\mathbf Y, z)$ is $L_S$-Lipschitz in $z$ uniformly over $\theta$, i.e. $\norm{S_\theta^{\mathbf Y}(z) - S_\theta^{\mathbf Y}(z')}\leq L_S \norm{z - z'}$. Next, at step $t$ of \eqref{eq:sMMLE} consider sampling from $p_{\theta_t}(\mathbf{Z}\mid \mathbf{Y})$ using \eqref{eq:FO} with step-size $\eta_t$ for $T^*(\eta_t)$ iterations where $T^*(\eta)$ is as prescribed in part 1 of Lemma~\ref{lem:fixedstepopt}. Assume its fixed-step hypotheses hold uniformly over $\theta_t$, with a uniform bound on the initial dual Wasserstein error and a common $B_0$, and that the mirror maps are uniformly $m_\varphi$-strongly convex for some $m_\varphi>0$. Set $B_{\mathrm{pr}}=B_0/m_\varphi$. Let $\mathbf{Z}_{t+1}$ be the sample obtained at the  $T^*(\eta_t)$-th transition of running sFO-MLA. Define
    $$\xi_{t+1} \coloneqq S_{\theta_t}^{\mathbf{Y}}(\mathbf{Z}_{t+1}) - \E[S_{\theta_t}^{\mathbf{Y}}(\mathbf{Z}_{t+1})\mid \mathcal{F}_t],$$
    where $\mathcal{F}_t$ is the natural filtration and $\E[\xi_{t+1}\mid \mathcal{F}_t] = 0$. Assume $\E[\norm{S_{\theta_t}^{\mathbf{Y}}(\mathbf{Z}_{t+1})}^2\mid \mathcal{F}_t] \leq C(1 + \norm{\theta^\star - \theta_t}^2)$ for some $C>0$. If 
    $$\sum \gamma_t = \infty, \,\, \sum \gamma_t^2 < \infty, \, \, \sum \gamma_t \eta_t < \infty,$$
    then $\theta_t \to \theta^\star$ almost surely. 
\end{theorem}
Theorem \ref{thm:mmle_convergence} provides practical guidance for the step-sizes $\gamma_t$ and $\eta_t$. For example, suppose $\gamma_t = \gamma_0/(t+1)^\alpha$ and $\eta_t = \eta_0/(t+1)^\beta$ for some $\alpha, \beta>0$, then for sequence \eqref{eq:sMMLE} to converge one needs $1/2< \alpha \leq 1$ and $\beta > 1- \alpha$.

Next we give a concrete example. Suppose we observe $y_i \in \{0, 1\}^d,  \, i = 1,\ldots, n$. To capture the correlation among the several binary outcomes it is common to assume that the data is generated by thresholding a latent Gaussian variable $z_i \in \mathbb{R}^d$, i.e. $y_{ij} = 1$ if and only if $z_{ij} \geq 0$. This is a generalization of the univariate probit model to the multivariate case \citep{ashford1970multi}. Thus, the correlation between $y_{ij}$ and $y_{ij'}$ can be studied through the correlation between $z_{ij}$ and $z_{ij'}$ where $j,j'\in\{1,\ldots,d\}$ and $j\neq j'$ . Suppose we assume that $z_i \sim \mathrm{N}(0,\Sigma)$. For identifiability, one typically assumes that $\Sigma$ is a correlation matrix, but that is not important for the ensuing development. Next, consider a scenario where $d$ is high and it is assumed that the dependence in $z$ is driven by $k\ll d$ unobserved latent factors $\zeta \in \mathbb{R}^{k \times 1}$ so that the residuals after removing the effect of $\zeta$ are independent. In other words, we assume
\begin{equation}\label{eq:factor_model}
    z_i = \Theta \zeta_i +\epsilon_i, \quad \Theta \in \mathbb{R}^{d \times k}, \, \, \epsilon_i \overset{iid}{\sim}\mathrm{N}(0, D), \, \, i = 1, \ldots, n, 
\end{equation}
where $D$ is a diagonal matrix and $\Theta$ is the factor-loading matrix. For simplicity, let us assume that $D$ is known. Thus, if we assume that $\text{Cov}(\zeta_i) = \mathrm{I}_k$ for all $i = 1, \ldots, n$, then $\Sigma$ admits the decomposition $\Sigma = \Theta \Theta' + D$. The inferential goal is to estimate the factor-loading matrix $\Theta$, which is in general not identifiable since for any orthogonal matrix $Q^{k \times k}$, $\Theta^\star = \Theta Q'$  satisfies $\Sigma = \Theta^\star \Theta^{\star'} + D$. A standard workaround is to assume that $\Theta$ is lower trapezoidal with positive diagonal entries. Hence, the parameter space is 
$$\{\Theta\in\mathbb R^{d\times k}:\Theta_{jj\prime}=0\text{ for }1\le j<j\prime\le k,\quad\Theta_{jj}>0\text{ for }1\le j\le k\}.$$

The observed likelihood, i.e. the likelihood for the binary observations, is a marginal over the latent $z$'s. For maximum likelihood estimation, we perform the following updates
\begin{equation}\label{eq:mmle_factor_model}
    \Theta_{t+1} = \Theta_t + \gamma_t \sum_{i=1}^n \E_{p(z_i, \zeta_i \mid y_i, \Theta_t)} [\nabla_\Theta \log f(y_i, z_i, \zeta_i \mid \Theta_t)],
\end{equation}
where $y$ is a deterministic function of $z$ and $f(y,z,\zeta\mid\Theta_t)=\mathbf1_{\mathcal C(y)}(z)f(z\mid\Theta_t,\zeta)q(\zeta)$. Furthermore, $p(z\mid y, \zeta, \Theta)$ is a constrained distribution where the constraints are determined by the observed $y$'s (orthants in this case), i.e. $z_i \mid y_i \in \mathcal{C}(y_i)$. Thus, to implement the iterations \eqref{eq:mmle_factor_model}, we need to sample from a constrained distribution where the constraints are linear. 

Commonly it is assumed that $q(\zeta)$ is $\mathrm{N}(0, \mathrm{I}_k)$ so that it can be easily marginalized out. Indeed, under this assumption, $z \mid \Theta_t \sim \mathrm{N}(0, \Theta_t \Theta_t' + D)$ and $p(z \mid y, \Theta_t) \propto \mathrm{N}(0, \Theta_t \Theta_t' + D) \mathbb{I}_{\mathcal{C}(y)}$. Hence in this case, although we have a constrained sampling problem, $-\log p(z\mid y, \Theta)$ or the potential of the target density is deterministically available. As a result, a standard implementation of MLA will be sufficient for this problem. The applicability of sFO-MLA starts to become apparent once we deviate from the two key assumptions of the model: a) Gaussianity of $\zeta$ and b) linear factor effects in \eqref{eq:factor_model}. If either assumption fails, then $p(z \mid y, \Theta)$ is not available deterministically. Before proceeding with the details of how sFO-MLA fits into the framework above, we briefly remark here that generalizing the model \eqref{eq:factor_model} by either assuming non-Gaussian $\zeta$ or non-linear factor effects is not artificial and in fact is well motivated. See for example \cite{wall2012mixture} for non-Gaussian $\zeta$ and \cite{cudeck2007factor} for non-linear factor models. In what follows, we adopt the case where $\zeta$ is Gaussian but the effect of the latent factors are non-linear. Specifically, the model we work with is
\begin{equation}
    z_i = g_\Theta(\zeta_i) + \epsilon_i, \quad \epsilon_i \sim \mathrm{N}(0, D), \zeta_i \sim \mathrm{N}(0, \mathrm{I}_k),
\end{equation}
with $\epsilon$ and $\zeta$ being independent of each other, and $g_\Theta(\zeta)$ is a non-linear function of $\zeta$ parameterized by $\Theta$. Then to implement \eqref{eq:mmle_factor_model}, we first note that if $q(\zeta)$ is independent of $\Theta$ and the complete data is $(y, z, \zeta)$, then 
$$\nabla_\Theta \log f(y, z,\zeta\mid\Theta) = J_\Theta g_\Theta(\zeta)' D^{-1} (z - g_\Theta(\zeta)),$$
where $J_\Theta$ is the Jacobian of $g_\Theta$ with respect to $\Theta$. Second, a sample from $(z, \zeta \mid y, \Theta)$ can be obtained by sampling from $z \mid y, \Theta$ and then sampling $\zeta\mid z, y, \Theta$. But 
$$p(z\mid y,\Theta)\propto\mathbf1_{\mathcal C(y)}(z)\int\mathrm N_d(z;g_\Theta(\zeta),D)q(\zeta)\,d\zeta.$$
which is no longer tractable. 
Now, applying Fisher's identity (this time treating $z$ to be the data and $\zeta$ to be the latent variable) we get
$$-\nabla_z \log p(z \mid y, \Theta) = \E[D^{-1} (z - g_\Theta(\zeta^\star))] = \E[G(z, \zeta^\star)], \quad \zeta^\star \sim p(\zeta \mid z, y, \Theta),$$
This gives an unbiased stochastic first-order oracle for the potential $-\log p(z\mid y,\Theta)$ in the interior of $\mathcal C(y)$, provided exact sampling from $p(\zeta\mid z,y,\Theta)$ is available. With approximate conditional samples, the resulting oracle can be biased. Moreover, $z \mid y, \Theta$ is supported on the set $\mathcal{C}(y)$. Thus, sFO-MLA is ideally suited for this task. Letting $Z_{t+1}$ be the sample generated by sFO-MLA at the $t$-th iteration from $p(z\mid y, \Theta_t)$, the updates of \eqref{eq:mmle_factor_model} are implemented as
\begin{equation}
    \Theta_{t+1} = \Theta_t + \gamma_t \sum_{i=1}^n  J_{\Theta_t} g_{\Theta_t}(\zeta_{i, t+1})' D^{-1}(Z_{i, t+1} - g_{\Theta_t}(\zeta_{i, t+1})).
\end{equation}

\subsection{Proof of Theorem \ref{thm:mmle_convergence}}
Define $g_t \coloneqq S_{\theta_t}^{\mathbf{Y}}(\mathbf{Z}_{t+1})$. Conditionally on $\mathcal{F}_t$, 
\begin{align*}
\E[g_t\mid\mathcal F_t]
&=\int S_{\theta_t}^{\mathbf Y}(z)\pi_{\theta_t}(dz\mid\mathbf Y)\\
&=\nabla\ell(\theta_t)+b_t,
\end{align*}
where $\pi_{\theta_t}(\cdot\mid\mathbf Y)$ is the conditional law of the
approximate sample produced by sFO-MLA, and
\[
b_t=\int S_{\theta_t}^{\mathbf Y}(z)\pi_{\theta_t}(dz\mid\mathbf Y)
-\int S_{\theta_t}^{\mathbf Y}(z)p_{\theta_t}(dz\mid\mathbf Y).
\]
Since sFO-MLA has been run for $T^*(\eta_t)$ iterations, we know 
$$\Wtwo\left(p_{\theta_t}(\cdot\mid \mathbf{Y}), \pi_{\theta_t}(\cdot \mid \mathbf{Y}) \right) \leq 2B_{\mathrm{pr}} \sqrt{\eta_t}$$
from Lemma \ref{lem:fixedstepopt}. Moreover, by the assumption of Lipschitz complete-data score and Kantorovich coupling we obtain
$$\norm{b_t} \leq L_S \Wtwo\left(p_{\theta_t}(\cdot\mid \mathbf{Y}), \pi_{\theta_t}(\cdot \mid \mathbf{Y}) \right) \leq 2B_{\mathrm{pr}} L_S\sqrt{\eta_t}.$$
Thus, $\norm{b_t}^2 \leq 4L_S^2 B_{\mathrm{pr}}^2 \eta_t$. Next consider the updates
\begin{align*}
    \theta_{t+1} = \theta_t + \gamma_t S_{\theta_t}^{\mathbf{Y}}(\mathbf{Z}_{t+1}) &= \theta_t + \gamma_t g_t\\
    & = \theta_t + \gamma_t\left[ \nabla \ell(\theta_t) + b_t + \xi_{t+1}\right].
\end{align*}
Let $e_t = \theta_t - \theta^\star$ and $V_t = \norm{e_t}^2$. Then the updates of \eqref{eq:sMMLE} can be alternatively written as $e_{t+1} = e_t + \gamma_t g_t$. Hence,
\begin{align*}
    V_{t+1} = V_t + 2\gamma_t \langle e_t, g_t \rangle + \gamma_t^2 \norm{g_t}^2.
\end{align*}
Using $g_t = \nabla \ell(\theta_t) + b_t + \xi_{t+1}$ and $\E[V_t \mid \mathcal{F}_t] = V_t$ we get
\begin{align*}
    \E[V_{t+1}\mid \mathcal{F}_t] = V_t + 2\gamma_t \langle e_t, \nabla \ell(\theta_t) \rangle + 2\gamma_t \langle e_t,  b_t \rangle + \gamma_t^2 \E[\norm{g_t}^2\mid \mathcal{F}_t]. 
\end{align*}
For $2\gamma_t \langle e_t, \nabla \ell(\theta_t) \rangle$, use the fact that there exists $m>0$ such that $\langle \theta - \theta^\star, \nabla \ell(\theta) \rangle \leq - m \norm{\theta - \theta^\star}^2$ for all $\theta$ to obtain $2\gamma_t \langle e_t, \nabla \ell(\theta_t) \rangle \leq - 2\gamma_t m \norm{e_t}^2 =  - 2\gamma_t m V_t$. By Young's inequality we get 
\begin{align*}
     2\gamma_t \langle e_t,  b_t \rangle \leq m \gamma_t V_t + \frac{\gamma_t}{m} \norm{b_t}^2 \leq m \gamma_t V_t + \frac{\gamma_t}{m} 4L_S^2 B_{\mathrm{pr}}^2 \eta_t.
\end{align*}
Finally, using $\gamma_t^2 \E[\norm{g_t}^2 \mid \mathcal{F}_t] \leq C\gamma_t^2(1 + V_t)$, we get
\begin{align*}
    \E[V_{t+1} \mid \mathcal{F}_t] \leq (1 + C\gamma_t^2) V_t -m \gamma_t V_t + C\gamma_t^2 + \frac{\gamma_t}{m} 4L_S^2 B_{\mathrm{pr}}^2 \eta_t.
\end{align*}
From \cite{robbins1971convergence} we thus get that $V_t$ converges almost surely to a finite random variable $V$ and $\sum_t \gamma_t V_t < \infty$. The proof is completed by observing that since $\sum \gamma_t = \infty$,  $V_t \geq 0$ and $V_t \to V$ almost surely, one must have $V = 0$ almost surely. 

\end{document}